\documentclass[12pt]{article}
\usepackage{fullpage,amsmath,amsthm,amssymb,comment,graphicx}
\usepackage[pdftex]{hyperref}
\hypersetup{
    linktocpage=true,
    colorlinks=true,				
    linkcolor=blue,		
    citecolor=green,	
    urlcolor=blue,	
}
\newtheorem{lem}{Lemma}[section]
\newtheorem{defn}[lem]{Definition}
\newtheorem{thm}[lem]{Theorem}

\newtheorem{prop}[lem]{Proposition}
\newtheorem{cor}[lem]{Corollary}

\newtheorem{rem}[lem]{Remark}
\newcommand{\eps}{\varepsilon}

\newcommand{\ex}[2]{\underset{#1}{\mathbb{E}}\left[#2\right]}
\newcommand{\pr}[2]{\underset{#1}{\mathbb{P}}\left[#2\right]}

\newcommand{\dr}[3]{\mathrm{D}_{#1}\left(#2\middle\|#3\right)}

\newcommand{\nope}[1]{}
\newcommand{\privloss}[2]{\mathsf{PrivLoss}\left(#1\middle\|#2\right)}

\newcommand{\con}{\xi}
\newcommand{\lin}{\rho}
\newcommand{\IIP}{zCDP}
\newcommand{\mcdp}{mCDP}
\newcommand{\cdp}{CDP}

\title{Concentrated Differential Privacy:\\Simplifications, Extensions, and Lower Bounds}
\author{Mark Bun\thanks{Supported by an NDSEG Fellowship and NSF grant CNS-1237235.} \and Thomas Steinke\thanks{Supported by NSF grants CCF-1116616, CCF-1420938, and CNS-1237235.}}
\date{\texttt{\{mbun,tsteinke\}@seas.harvard.edu}}
\begin{document}
\maketitle

\begin{abstract}
``Concentrated differential privacy'' was recently introduced by Dwork and Rothblum as a relaxation of differential privacy, which permits sharper analyses of many privacy-preserving computations.
We present an alternative formulation of the concept of concentrated differential privacy in terms of the R\'{e}nyi divergence between the distributions obtained by running an algorithm on neighboring inputs. With this reformulation in hand, we prove sharper quantitative results, establish lower bounds, and raise a few new questions. We also unify this approach with approximate differential privacy by giving an appropriate definition of ``approximate concentrated differential privacy.''
\end{abstract}

\break
{\small
\tableofcontents
}
\break 

\section{Introduction}

Differential privacy \cite{DworkMNS06} is a formal mathematical standard for protecting individual-level privacy in statistical data analysis. In its simplest form, (pure) differential privacy is parameterized by a real number $\eps > 0$, which controls how much ``privacy loss''\footnote{The privacy loss is a random variable which quantifies how much information is revealed about an individual by a computation involving their data; it depends on the outcome of the computation, the way the computation was performed, and the information that the individual wants to hide. We discuss it informally in this introduction and define it precisely in Definition \ref{defn:PrivLoss} on page \pageref{defn:PrivLoss}.} an individual can suffer when a computation (i.e., a statistical data analysis task) is performed involving his or her data. 

One particular hallmark of differential privacy is that it degrades smoothly and predictably under the \emph{composition} of multiple computations. In particular, if one performs $k$ computational tasks that are each $\eps$-differentially private and combines the results of those tasks, then the computation as a whole is $k\eps$-differentially private. This property makes differential privacy amenable to the type of modular reasoning used in the design and analysis of algorithms: When a sophisticated algorithm is comprised of a sequence of differentially private steps, one can establish that the algorithm as a whole remains differentially private.

A widely-used relaxation of pure differential privacy is \emph{approximate} or $(\eps, \delta)$-differential privacy \cite{DworkKMMN06}, which essentially guarantees that the probability that any individual suffers privacy loss exceeding $\eps$ is bounded by $\delta$. For sufficiently small $\delta$, approximate $(\eps,\delta)$-differential privacy provides a comparable standard of privacy protection as pure $\eps$-differential privacy, while often permitting substantially more useful analyses to be performed.

Unfortunately, there are situations where, unlike pure differential privacy, approximate differential privacy is not a very elegant abstraction for mathematical analysis, particularly the analysis of composition. The ``advanced composition theorem'' of Dwork, Rothblum, and Vadhan \cite{DworkRV10} (subsequently improved by \cite{KairouzOV15,MurtaghV16}) shows that the composition of $k$ tasks which are each $(\eps, \delta)$-differentially private is $ (\approx\!\!\sqrt{k}\eps, \approx\!\!k\delta)$-differentially private. However, these bounds can be unwieldy; computing the tightest possible privacy guarantee for the composition of $k$ arbitrary mechanisms with differing $(\eps_i, \delta_i)$-differential privacy guarantees is $\#\mathsf{P}$-hard \cite{MurtaghV16}! Furthermore, these bounds are not tight even for simple and natural privacy-preserving computations. For instance, consider the mechanism which approximately answers $k$ statistical queries on a given database by adding independent Gaussian noise to each answer. Even for this basic computation, the advanced composition theorem does not yield a tight analysis.\footnote{In particular, consider answering $k$ statistical queries on a dataset of $n$ individuals by adding noise drawn from $\mathcal{N}(0,(\sigma/n)^2)$ independently for each query. Each individual query satisfies $(O(\sqrt{\log(1/\delta)}/\sigma),\delta)$-differential privacy for any $\delta>0$. Applying the advanced composition theorem shows that the composition of all $k$ queries satisfies $(O(\sqrt{k} \log(1/\delta)/\sigma),(k+1)\delta)$-differential privacy for any $\delta>0$. However, it is well-known that this bound can be improved to $(O(\sqrt{k \log(1/\delta)}/\sigma),\delta)$-differential privacy.}

Dwork and Rothblum \cite{DworkR16} recently put forth a different relaxation of differential privacy called \emph{concentrated differential privacy}. Roughly, a randomized mechanism satisfies concentrated differentially privacy if the privacy loss has small mean and is subgaussian. Concentrated differential privacy behaves in a qualitatively similar way as approximate $(\eps, \delta)$-differential privacy under composition. However, it permits sharper analyses of basic computational tasks, including a tight analysis of the aforementioned Gaussian mechanism.

Using the work of Dwork and Rothblum \cite{DworkR16} as a starting point, we introduce an alternative formulation of the concept of concentrated differential privacy that we call ``zero-concentrated differential privacy'' (\IIP{} for short). To distinguish our definition from that of Dwork and Rothblum, we refer to their definition as ``mean-concentrated differential privacy'' (\mcdp{} for short). Our definition uses the R\'{e}nyi divergence between probability distributions as a different method of capturing the requirement that the privacy loss random variable is subgaussian.

\subsection{Our Reformulation: Zero-Concentrated Differential Privacy}

As is typical in the literature, we model a dataset as a multiset or tuple of $n$ elements (or ``rows'') in $\mathcal{X}^n$, for some ``data universe'' $\mathcal{X}$, where each element represents one individual's information. A (privacy-preserving) computation is a randomized algorithm $M : \mathcal{X}^n \to \mathcal{Y}$, where $\mathcal{Y}$ represents the space of all possible outcomes of the computation.

\begin{defn}[Zero-Concentrated Differential Privacy (\IIP{})]
A randomised mechanism $M : \mathcal{X}^n \to \mathcal{Y}$ is $(\con,\lin)$-zero-concentrated differentially private (henceforth $(\con,\lin)$-\IIP{}) if, for all $x,x' \in \mathcal{X}^n$ differing on a single entry and all $\alpha \in (1,\infty)$, \begin{equation}\dr{\alpha}{M(x)}{M(x')} \leq \con + \lin \alpha,\label{eqn:IIP-Renyi} \end{equation} where $\dr{\alpha}{M(x)}{M(x')}$ is the $\alpha$-R\'enyi divergence\footnote{R\'enyi divergence has a parameter $\alpha \in (1,\infty)$ which allows it to interpolate between KL-divergence ($\alpha \!\to\! 1$) and max-divergence ($\alpha \!\to\! \infty$). It should be thought of as a measure of dissimilarity between distributions. We define it formally in Section \ref{sec:Renyi}.   Throughout, we assume that all logarithms are natural unless specified otherwise --- that is, base $e \approx 2.718$. This includes logarithms in information theoretic quantities like entropy, divergence, and mutual information, whence these quantities are measured in \emph{nats} rather than in \emph{bits}.}  between the distribution of $M(x)$ and the distribution of $M(x')$.

We define $\lin$-\IIP{} to be $(0,\lin)$-\IIP{}.\footnote{For clarity of exposition, we consider only $\lin$-\IIP{} in the introduction and give more general statements for $(\con,\lin)$-\IIP{} later. We also believe that having a one-parameter definition is desirable.}
\end{defn}
Equivalently, we can replace \eqref{eqn:IIP-Renyi} with \begin{equation}\ex{}{e^{(\alpha-1)Z}} \leq e^{(\alpha-1)(\con + \lin \alpha)},\label{eqn:GDP-MGF}\end{equation} where $Z = \privloss{M(x)}{M(x')}$ is the privacy loss random variable:
\begin{defn}[Privacy Loss Random Variable] \label{defn:PrivLoss}
Let $Y$ and $Y'$ be random variables on $\Omega$. We define the \emph{privacy loss random variable between $Y$ and $Y'$} -- denoted $Z=\privloss{Y}{Y'}$ -- as follows. Define a function $f : \Omega \to \mathbb{R}$ by $f(y) = \log(\pr{}{Y=y}/\pr{}{Y'=y})$.\footnote{Throughout we abuse notation by letting $\pr{}{Y=y}$ represent either the probability mass function or the probability density function of $Y$ evaluated at $y$. Formally, $\pr{}{Y=y}/\pr{}{Y'=y}$ denotes the Radon-Nikodym derivative of the measure $Y$ with respect to the measure $Y'$ evaluated at $y$, where we require $Y$ to be absolutely continuous with respect to $Y'$, i.e. $Y \ll Y'$.} Then $Z$ is distributed according to $f(Y)$.
\end{defn}
Intuitively, the value of the privacy loss $Z = \privloss{M(x)}{M(x')}$ represents how well we can distinguish $x$ from $x'$ given only the output $M(x)$ or $M(x')$. If $Z>0$, then the observed output of $M$ is more likely to have occurred if the input was $x$ than if $x'$ was the input. Moreover, the larger $Z$ is, the bigger this likelihood ratio is. Likewise, $Z<0$ indicates that the output is more likely if $x'$ is the input. If $Z=0$, both $x$ and $x'$ ``explain'' the output of $M$ equally well.

A mechanism $M : \mathcal{X}^n \to \mathcal{Y}$ is $\varepsilon$-differentially private if and only if $\pr{}{Z>\varepsilon}=0$, where $Z = \privloss{M(x)}{M(x')}$ is the privacy loss of $M$ on arbitrary inputs $x,x' \in \mathcal{X}^n$ differing in one entry.
On the other hand, $M$ being $(\varepsilon,\delta)$-differentially private is equivalent, up to a small loss in parameters, to the requirement that $\pr{}{Z>\varepsilon}\leq \delta$.

In contrast, \IIP{} entails a bound on the \emph{moment generating function} of the privacy loss $Z$ --- that is, $\ex{}{e^{(\alpha-1)Z}}$ as a function of $\alpha-1$. The bound \eqref{eqn:GDP-MGF} implies that $Z$ is a \emph{subgaussian} random variable\footnote{A random variable $X$ being subgaussian is characterised by the following four equivalent conditions \cite{Rivasplata12}. (i) $\pr{}{|X-\ex{}{X}|>\lambda} \leq e^{-\Omega(\lambda^2)}$ for all $\lambda>0$. (ii) $\ex{}{e^{t(X-\ex{}{X})}} \leq e^{O(t^2)}$ for all $t \in \mathbb{R}$. (iii) $\ex{}{(X-\ex{}{X})^{2k}} \leq O(k)^k$ for all $k \in \mathbb{N}$. (iv) $\ex{}{e^{c(X-\ex{}{X})^2}} \leq 2$ for some $c>0$.} with small mean. Intuitively, this means that $Z$ resembles a Gaussian distribution with mean $\con+\lin$ and variance $2\lin$. In particular, we obtain strong tail bounds on $Z$. Namely \eqref{eqn:GDP-MGF} implies that $$\pr{}{Z>\lambda+\con+\lin} \leq e^{-\lambda^2/4\lin}$$ for all $\lambda>0$.\footnote{We only discuss bounds on the upper tail of $Z$. We can obtain similar bounds on the lower tail of $Z=\privloss{M(x)}{M(x')}$ by considering $Z'=\privloss{M(x')}{M(x)}$.}

Thus \IIP{} requires that the privacy loss random variable is concentrated around zero (hence the name). That is, $Z$ is ``small'' with high probability, with larger deviations from zero becoming increasingly unlikely. Hence we are unlikely to be able to distinguish $x$ from $x'$ given the output of $M(x)$ or $M(x')$. Note that the randomness of the privacy loss random variable is taken only over the randomnesss of the mechanism $M$.

\subsubsection{Comparison to the Definition of Dwork and Rothblum}

For comparison, Dwork and Rothblum \cite{DworkR16} define \emph{$(\mu,\tau)$-concentrated differential privacy} for a randomized mechanism $M : \mathcal{X}^n \to \mathcal{Y}$ as the requirement that, if $Z=\privloss{M(x)}{M(x')}$ is the privacy loss for $x,x'\in \mathcal{X}^n$ differing on one entry, then $$\ex{}{Z} \leq \mu \qquad\text{and}\qquad \ex{}{e^{(\alpha-1)(Z-\ex{}{Z})}} \leq e^{(\alpha-1)^2\frac12\tau^2}$$ for all $\alpha \in \mathbb{R}$. That is, they require both a bound on the mean of the privacy loss and that the privacy loss is tightly concentrated around its mean.
To distinguish our definitions, we refer to their definition as \emph{mean-concentrated differential privacy} (or \mcdp{}).

Our definition, \IIP{}, is a \emph{relaxation} of \mcdp{}. In particular, a $(\mu,\tau)$-\mcdp{} mechanism is also $(\mu-\tau^2/2,\tau^2/2)$-\IIP{} (which is tight for the Gaussian mechanism example), whereas the converse is not true. (However, a partial converse holds; see Lemma \ref{lem:IIPtoCDP}.) 

\subsection{Results}

\subsubsection{Relationship between \IIP{} and Differential Privacy} \label{sec:GDP-DP-relationship}

Like Dwork and Rothblum's formulation of concentrated differential privacy, \IIP{} can be thought of as providing guarantees of $(\eps, \delta)$-differential privacy \emph{for all} values of $\delta > 0$:

\begin{prop} \label{prop:CDPtoDP-intro}
If $M$ provides $\lin$-\IIP{}, then $M$ is $(\lin+2\sqrt{\lin\log(1/\delta)},\delta)$-differentially private for any $\delta>0$.
\end{prop}
We also prove a slight strengthening of this result (Lemma \ref{lem:CDPtoDP2}). Moreover, there is a partial converse, which shows that, up to a loss in parameters, \IIP{} is equivalent to differential privacy with this $\forall \delta>0$ quantification (see Lemma \ref{lem:DPtoCDP}).

There is also a direct link from pure differential privacy to \IIP{}:

\begin{prop} \label{prop:PDPtoCDP-intro}
If $M$ satisfies $\eps$-differential privacy, then $M$ satisfies $(\frac12 \eps^2)$-\IIP{}.
\end{prop}

Dwork and Rothblum \cite[Theorem 3.5]{DworkR16} give a slightly weaker version of Proposition  \ref{prop:PDPtoCDP-intro}, which implies that $\eps$-differential privacy yields $(\frac12 \eps (e^\eps-1))$-\IIP{}; this improves on an earlier bound \cite{DworkRV10} by the factor $\frac12$.

We give proofs of these and other properties using properties of R\'{e}nyi divergence in Sections \ref{sec:Renyi} and \ref{sec:DP}.

Propositions \ref{prop:CDPtoDP-intro} and \ref{prop:PDPtoCDP-intro} show that \IIP{} is an intermediate notion between pure differential privacy and approximate differential privacy. Indeed, many algorithms satisfying approximate differential privacy do in fact also satisfy \IIP{}.

\subsubsection{Gaussian Mechanism} \label{sec:Gaussian-intro}

Just as with \mcdp{}, the prototypical example of a mechanism satisfying \IIP{} is the \emph{Gaussian mechanism}, which answers a real-valued query on a database by perturbing the true answer with Gaussian noise.

\begin{defn}[Sensitivity]
A function $q : \mathcal{X}^n \to \mathbb{R}$ has \emph{sensitivity} $\Delta$ if for all $x, x' \in \mathcal{X}^n$ differing in a single entry, we have $|q(x) - q(x')| \le \Delta$.
\end{defn}

\begin{prop}[Gaussian Mechanism] \label{prop:gaussian-mech}
Let $q : \mathcal{X}^n \to \mathbb{R}$ be a sensitivity-$\Delta$ query. Consider the mechanism $M : \mathcal{X}^n \to \mathbb{R}$ that on input $x$, releases a sample from $\mathcal{N}(q(x), \sigma^2)$. Then $M$ satisfies $(\Delta^2/2\sigma^2)$-\IIP{}.
\end{prop}

We remark that either inequality defining \IIP{} --- \eqref{eqn:IIP-Renyi} or \eqref{eqn:GDP-MGF} --- is exactly tight for the Gaussian mechanism for all values of $\alpha$. Thus the definition of \IIP{} seems tailored to the Gaussian mechanism.

\subsubsection{Basic Properties of \IIP} \label{sec:BasicProperties}
Our definition of \IIP{} satisfies the key basic properties of differential privacy. Foremost, these properties include smooth degradation under composition, and invariance under postprocessing:

\begin{lem}[Composition] \label{lem:Composition-intro}
Let $M : \mathcal{X}^n \to \mathcal{Y}$ and $M' : \mathcal{X}^n \to \mathcal{Z}$ be randomized algorithms. Suppose $M$ satisfies $\lin$-\IIP{} and $M'$ satisfies $\lin'$-\IIP{}. Define $M'' : \mathcal{X}^n \to \mathcal{Y} \times \mathcal{Z}$ by $M''(x) = (M(x),M'(x))$. Then $M''$ satisfies $(\lin+\lin')$-\IIP{}.
\end{lem}
\begin{lem}[Postprocessing] \label{lem:Postprocessing-intro}
Let $M : \mathcal{X}^n \to \mathcal{Y}$ and $f : \mathcal{Y} \to \mathcal{Z}$ be randomized algorithms. Suppose $M$ satisfies $\lin$-\IIP{}. Define $M' : \mathcal{X}^n \to \mathcal{Z}$ by $M'(x) = f(M(x))$. Then $M'$ satisfies $\lin$-\IIP{}.
\end{lem}
These properties follow immediately from corresponding properties  of the R\'enyi divergence outlined in Lemma \ref{lem:Renyi}.

We remark that Dwork and Rothblum's definition of \mcdp{} is not closed under postprocessing; we provide a counterexample in Appendix \ref{app:pp-mcdp}. (However, an arbitrary amount of postprocessing can worsen the guarantees of \mcdp{} by at most constant factors.)

\subsubsection{Group Privacy} \label{sec:IntroGroupPrivacy}

A mechanism $M$ guarantees \emph{group privacy} if no small group of individuals has a significant effect on the outcome of a computation (whereas the definition of \IIP{} only refers to individuals, which are groups of size $1$). That is, group privacy for groups of size $k$ guarantees that, if $x$ and $x'$ are inputs differing on $k$ entries (rather than a single entry), then the outputs $M(x)$ and $M(x')$ are close. 

Dwork and Rothblum \cite[Theorem 4.1]{DworkR16} gave nearly tight bounds on the group privacy guarantees of concentrated differential privacy, showing that a $(\mu=\tau^2/2, \tau)$-concentrated differentially private mechanism affords $(k^2\mu \cdot (1 + o(1)), k\tau \cdot (1 + o(1)))$-concentrated differential privacy for groups of size $k = o(1/\tau)$. We are able to show a group privacy guarantee for \IIP{} that is exactly tight and works for a wider range of parameters:

\begin{prop} \label{prop:group-intro}
Let $M : \mathcal{X}^n \to \mathcal{Y}$ satisfy $\lin$-\IIP{}. Then $M$ guarantees $(k^2\lin)$-\IIP{} for groups of size $k$ --- i.e. for every $x, x' \in \mathcal{X}^n$ differing in up to $k$ entries and every $\alpha \in (1, \infty)$, we have
\[\dr{\alpha}{M(x)}{M(x')} \le (k^2\lin) \cdot \alpha.\]
\end{prop}

In particular, this bound is achieved (simultaneously for all values $\alpha$) by the Gaussian mechanism. Our proof is also simpler than that of Dwork and Rothblum; see Section \ref{sec:GroupPrivacy}.

\subsubsection{Lower Bounds}

The strong group privacy guarantees of \IIP{} yield, as an unfortunate consequence, strong lower bounds as well. We show that, as with pure differential privacy, \IIP{} is susceptible to  information-based lower bounds, as well as to so-called packing arguments \cite{HardtT10,McGregorMPRTV10, De12}:

\begin{thm} \label{thm:LB-intro}
Let $M : \mathcal{X}^n \to \mathcal{Y}$ satisfy $\lin$-\IIP{}. Let $X$ be a random variable on $\mathcal{X}^n$. Then $$I\left(X;M(X)\right) \leq \lin \cdot n^2,$$ where $I(\cdot;\cdot)$ denotes the mutual information between the random variables (in nats, rather than bits). Furthermore, if the entries of $X$ are independent, then $I(X;M(X)) \leq \lin \cdot n$.
\end{thm}

Theorem \ref{thm:LB-intro} yields strong lower bounds for \IIP{} mechanisms, as we can construct distributions $X$ such that, for any accurate mechanism $M$, $M(X)$ reveals a lot of information about $X$ (i.e.~$I(X;M(X))$ is large for any accurate $M$).

In particular, we obtain a strong separation between approximate differential privacy and \IIP{}. For example, we can show that releasing an accurate approximate histogram (or, equivalently, accurately answering all point queries) on a data domain of size $k$ requires an input with at least $n=\Theta(\sqrt{\log k})$ entries to satisfy \IIP{}. In contrast, under approximate differential privacy, $n$ can be \emph{independent} of the domain size $k$ \cite{BeimelNS13}!  In particular, our lower bounds show that ``stability-based'' techniques (such as those in the propose-test-release framework \cite{DworkL09}) are not compatible with \IIP{}. 

Our lower bound exploits the strong group privacy guarantee afforded by \IIP{}. Group privacy has been used to prove tight lower bounds for pure differential privacy \cite{HardtT10,De12} and approximate differential privacy \cite{SteinkeU15b}. These results highlight the fact that group privacy is often the limiting factor for private data analysis. For $(\varepsilon,\delta)$-differential privacy, group privacy becomes vacuous for groups of size $k=\Theta(\log(1/\delta)/\varepsilon)$. Indeed, stability-based techniques exploit precisely this breakdown in group privacy.

As a result of this strong lower bound, we show that any mechanism for answering statistical queries that satisfies \IIP{} can be converted into a mechanism satisfying pure differential privacy with only a quadratic blowup in its sample complexity. More precisely, the following theorem illustrates a more general result we prove in Section \ref{sec:CDPvsPDP}.

\begin{thm} Let $n \in \mathbb{N}$ and $\alpha \geq 1/n$ be arbitrary. Set $\eps=\alpha$ and $\lin=\alpha^2$.
Let $q : \mathcal{X} \to [0,1]^k$ be an arbitrary family of statistical queries.
Suppose $M : \mathcal{X}^n \to [0,1]^k$ satisfies $\lin$-\IIP{} and $$\ex{M}{\|M(x)-q(x)\|_\infty} \leq \alpha$$ for all $x \in \mathcal{X}^n$. 
Then there exists $M' : \mathcal{X}^{n'} \to [0,1]^k$ for $n'=5n^2$ satisfying $\varepsilon$-differential privacy and $$\ex{M'}{\|M'(x)-q(x)\|_\infty} \leq 10\alpha$$ for all $x \in \mathcal{X}^{n'}$.
\end{thm}

For some classes of queries, this reduction is essentially tight. For example, for $k$ one-way marginals, the Gaussian mechanism achieves sample complexity $n=\Theta(\sqrt{k})$ subject to \IIP{}, whereas the Laplace mechanism achieves sample complexity $n=\Theta(k)$ subject to pure differential privacy, which is known to be optimal.

For more details, see Sections \ref{sec:LowerBounds} and \ref{sec:CDPvsPDP}.

\subsubsection{Approximate \IIP{}}

To circumvent these strong lower bounds for \IIP{}, we consider a relaxation of \IIP{} in the spirit of approximate differential privacy that permits a small probability $\delta$ of (catastrophic) failure:

\begin{defn}[Approximate Zero-Concentrated Differential Privacy (Approximate \IIP{})]
A randomized mechanism $M : \mathcal{X}^n \to \mathcal{Y}$ is $\delta$-approximately $(\con,\lin)$-\IIP{} if, for all $x,x' \in \mathcal{X}^n$ differing on a single entry, there exist events $E$ (depending on $M(x)$) and $E'$ (depending on $M(x')$) such that $\pr{}{E} \geq 1-\delta$, $\pr{}{E'} \geq 1-\delta$, and $$\forall \alpha \in (1,\infty) \qquad \dr{\alpha}{M(x)|_{E}}{M(x')|_{E'}} \leq \con + \lin \cdot \alpha \qquad \wedge \qquad \dr{\alpha}{M(x')|_{E'}}{M(x)|_{E}} \leq \con + \lin \cdot \alpha,$$ where $M(x)|_{E}$ denotes the distribution of $M(x)$ conditioned on the event $E$. We further define $\delta$-approximate $\lin$-\IIP{} to be $\delta$-approximate $(0,\lin)$-\IIP{}.
\end{defn}

In particular, setting $\delta=0$ gives the original definition of \IIP{}. However, this definition unifies \IIP{} with approximate differential privacy:

\begin{prop}
If $M$ satisfies $(\varepsilon,\delta)$-differential privacy, then $M$ satisfies $\delta$-approximate $\frac12 \varepsilon^2$-\IIP{}.
\end{prop}

Approximate \IIP{} retains most of the desirable properties of \IIP{}, but allows us to incorporate stability-based techniques and bypass the above lower bounds. This also presents a unified tool to analyse a composition of \IIP{} with approximate differential privacy; see Section \ref{sec:ApproxCDP}.

\subsection{Related Work}

Our work builds on the aforementioned prior work of Dwork and Rothblum \cite{DworkR16}.\footnote{Although Dwork and Rothblum's work only appeared publicly in March 2016, they shared a preliminary draft of their paper with us before we commenced this work. As such, our ideas are heavily inspired by theirs.} We view our definition of concentrated differential privacy as being ``morally equivalent'' to their definition of concentrated differential privacy, in the sense that both definitions formalize the same concept.\footnote{We refer to our definition as ``zero-concentrated differential privacy'' (\IIP{}) and their definition as ``mean-concentrated differential privacy'' (\mcdp{}). We use ``concentrated differential privacy'' (\cdp{}) to refer to the underlying \emph{concept} formalized by both definitions.} (The formal relationship between the two definitions is discussed in Section \ref{sec:CDP-DR}.) 
However, the definition of \IIP{} generally seems to be easier to work with than that of \mcdp{}. In particular, our formulation in terms of R\'{e}nyi divergence simplifies many analyses.

Dwork and Rothblum prove several results about concentrated differential privacy that are similar to ours. Namely, they prove analogous properties of \mcdp{} as we prove for \IIP{} (cf.~Sections \ref{sec:GDP-DP-relationship}, \ref{sec:Gaussian-intro}, \ref{sec:BasicProperties}, and \ref{sec:IntroGroupPrivacy}). However, as noted, some of their bounds are weaker than ours; also, they do not explore lower bounds.

Several of the ideas underlying concentrated differential privacy are implicit in earlier works. In particular, the proof of the advanced composition theorem of Dwork, Rothblum, and Vadhan \cite{DworkRV10} essentially uses the ideas of concentrated differential privacy. Their proof contains analogs of Propositions \ref{lem:Composition-intro}, \ref{prop:CDPtoDP-intro}, and \ref{prop:PDPtoCDP-intro},. 

We also remark that Tardos \cite{Tardos03} used R\'{e}nyi divergence to prove lower bounds for cryptographic objects called \emph{fingerprinting codes}. Fingerprinting codes turn out to be closely related to differential privacy \cite{Ullman13,BunUV14,SteinkeU14}, and Tardos' lower bound can be (loosely) viewed as a kind of privacy-preserving algorithm.

\subsection{Further Work}

We believe that concentrated differential privacy is a useful tool for analysing private computations, as it provides both simpler and tighter bounds. We hope that \cdp{} will be prove useful in both the theory and  practice of differential privacy.

Furthermore, our lower bounds show that \cdp{} can really be a much more stringent condition than approximate differential privacy. Thus \cdp{} defines a ``subclass'' of all $(\eps, \delta)$-differentially private algorithms. This subclass includes most differentially private algorithms in the literature, but not all --- the most notable exceptions being algorithms that use the propose-test-release approach \cite{DworkL09} to exploit low local sensitivity.

This ``\cdp{} subclass'' warrants further exploration. In particular, is there a ``complete'' mechanism for this class of algorithms, in the same sense that the exponential mechanism \cite{McSherryT07,BlumLR08} is complete for pure differential privacy? Can we obtain a simple characterization of the sample complexity needed to satisfy \cdp{}? The ability to prove stronger and simpler lower bounds for \cdp{} than for approximate DP may be useful for showing the limitations of certain algorithmic paradigms. For example, any differentially private algorithm that only uses the Laplace mechanism, the exponential mechanism, the Gaussian mechanism, and the ``sparse vector'' technique, along with composition and postprocessing will be subject to the lower bounds for \cdp{}.

There is also room to examine how to interpret the \IIP{} privacy guarantee. In particular, we leave it as an open question to understand the extent to which $\lin$-\IIP{} provides a stronger privacy guarantee than the implied $(\varepsilon,\delta)$-DP guarantees (cf.~Proposition \ref{prop:CDPtoDP-intro}). 

In general, much of the literature on differential privacy can be re-examined through the lens of \cdp{}, which may yield new insights and results.

\section{R\'enyi Divergence} \label{sec:Renyi}
Recall the definition of R\'enyi divergence:
\begin{defn}[{R\'enyi Divergence \cite[Equation (3.3)]{Renyi61}}]
Let $P$ and $Q$ be probability distributions on $\Omega$. For $\alpha \in  (1,\infty)$, we define the \emph{R\'enyi divergence of order $\alpha$ between $P$ and $Q$} as 
\begin{align*}
\dr{\alpha}{P}{Q} =& \frac{1}{\alpha-1} \log \left( \int_\Omega P(x)^\alpha Q(x)^{1-\alpha} \mathrm{d} x \right) \\=& \frac{1}{\alpha-1} \log \left( \ex{x \sim Q}{ \left(\frac{P(x)}{Q(x)}\right)^\alpha} \right) \\=& \frac{1}{\alpha-1} \log \left( \ex{x \sim P}{ \left(\frac{P(x)}{Q(x)}\right)^{\alpha-1}} \right),
\end{align*}
where $P(\cdot)$ and $Q(\cdot)$ are the probability mass/density functions of $P$ and $Q$ respectively or, more generally, $P(\cdot)/Q(\cdot)$ is the Radon-Nikodym derivative of $P$ with respect to $Q$.\footnote{If $P$ is not absolutely continuous with respect to $Q$ (i.e. it is not the case that $P \ll Q$), we define $\dr{\alpha}{P}{Q}=\infty$ for all $\alpha \in [1,\infty]$.}
We also define the KL-divergence $$\dr{1}{P}{Q} = \lim_{\alpha \to 1} \dr{\alpha}{P}{Q} = \int_\Omega P(x) \log \left(\frac{P(x)}{Q(x)}\right) \mathrm{d}x$$ and the max-divergence $$\dr{\infty}{P}{Q} = \lim_{\alpha \to \infty} \dr{\alpha}{P}{Q} = \sup_{x \in \Omega} \log \left(\frac{P(x)}{Q(x)}\right) .$$
\end{defn}

Alternatively, R\'enyi divergence can be defined in terms of the privacy loss (Definition \ref{defn:PrivLoss}) between $P$ and $Q$: $$e^{(\alpha-1)\dr{\alpha}{P}{Q}} = \ex{Z\sim\privloss{P}{Q}}{e^{(\alpha-1)Z}}$$ for all $\alpha \in (1,\infty)$. Moreover, $\dr{1}{P}{Q} = \ex{Z\sim\privloss{P}{Q}}{Z}$.

We record several useful and well-known properties of R\'enyi divergence. We refer the reader to \cite{vanErvenH14} for proofs and discussion of these (and many other) properties. Self-contained proofs are given in Appendix \ref{app:Renyi}.
\begin{lem} \label{lem:Renyi}
Let $P$ and $Q$ be probability distributions and $\alpha \in [1,\infty]$.
\begin{itemize}
\item \emph{Non-negativity:} $\dr{\alpha}{P}{Q} \geq 0$ with equality if and only if $P=Q$.

\item \emph{Composition:} Suppose $P$ and $Q$ are distributions on $\Omega \times \Theta$. Let $P'$ and $Q'$ denote the marginal distributions on $\Omega$ induced by $P$ and $Q$ respectively. For $x \in \Omega$, let $P'_x$ and $Q'_x$ denote the conditional distributions on $\Theta$ induced by $P$ and $Q$ respectively, where $x$ specifies the first coordinate. Then $$\dr{\alpha}{P'}{Q'} + \min_{x \in \Omega} \dr{\alpha}{P'_x}{Q'_x} \leq \dr{\alpha}{P}{Q} \leq \dr{\alpha}{P'}{Q'} + \max_{x \in \Omega} \dr{\alpha}{P'_x}{Q'_x}.$$
In particular if $P$ and $Q$ are product distributions, then the R\'enyi divergence between $P$ and $Q$ is just the sum of the R\'enyi divergences of the marginals.

\item \emph{Quasi-Convexity:} Let $P_0, P_1$ and $Q_0, Q_1$ be distributions on $\Omega$, and let $P = tP_0 + (1-t)P_1$ and $Q = tQ_0 + (1-t)Q_1$ for $t \in [0, 1]$. Then $\dr{\alpha}{P}{Q} \le \max\{\dr{\alpha}{P_0}{Q_0}, \dr{\alpha}{P_1}{Q_1}\}$. Moreover, KL divergence is convex: $$\dr{1}{P}{Q} \leq t \dr{1}{P_0}{Q_0}+(1-t)\dr{1}{P_1}{Q_1}.$$

\item \emph{Postprocessing:} Let $P$ and $Q$ be distributions on $\Omega$ and let $f : \Omega \to \Theta$ be a function. Let $f(P)$ and $f(Q)$ denote the distributions on $\Theta$ induced by applying $f$ to $P$ or $Q$ respectively. Then $\dr{\alpha}{f(P)}{f(Q)} \leq \dr{\alpha}{P}{Q}$.

Note that quasi-convexity allows us to extend this guarantee to the case where $f$ is a randomized mapping.

\item \emph{Monotonicity:} For $1 \leq \alpha \leq \alpha' \leq \infty$, $\dr{\alpha}{P}{Q} \leq \dr{\alpha'}{P}{Q}$.

\end{itemize}
\end{lem}

\subsection{Composition and Postprocessing}

The following lemma gives the postprocessing and (adaptive) composition bounds (extending Lemmas \ref{lem:Composition-intro} and \ref{lem:Postprocessing-intro}).

\begin{lem}[Composition \& Postprocessing] \label{lem:CompositionPostprocessing}
Let $M : \mathcal{X}^n \to \mathcal{Y}$ and $M' : \mathcal{X}^n \times \mathcal{Y} \to \mathcal{Z}$. Suppose $M$ satisfies $(\con,\lin)$-\IIP{} and $M'$ satisfies $(\con',\lin')$-\IIP{} (as a function of its first argument). Define $M'' : \mathcal{X}^n \to \mathcal{Z}$ by $M''(x) = M'(x,M(x))$. Then $M''$ satisfies $(\con+\con',\lin+\lin')$-\IIP{}.
\end{lem}
The proof is immediate from Lemma \ref{lem:Renyi}. Note that, while Lemma \ref{lem:CompositionPostprocessing} is only stated for the composition of two mechanisms, it can be inductively applied to analyse the composition of arbitrarily many mechanisms.

\subsection{Gaussian Mechanism}

The following lemma gives the R\'enyi divergence between two Gaussian distributions with the same variance.

\begin{lem} \label{lem:NormalDivergence}
Let $\mu, \nu, \sigma \in \mathbb{R}$ and $\alpha \in [1,\infty)$. Then $$\dr{\alpha}{\mathcal{N}(\mu,\sigma^2)}{\mathcal{N}(\nu,\sigma^2)} = \frac{\alpha(\mu - \nu)^2}{2\sigma^2}$$
\end{lem}

Consequently, the Gaussian mechanism, which answers a sensitivity-$\Delta$ query by adding noise drawn from $\mathcal{N}(0,\sigma^2)$, satisfies $\left(\frac{\Delta^2}{2\sigma^2}\right)$-\IIP{} (Proposition \ref{prop:gaussian-mech}).

\begin{proof}
We calculate
\begin{align*}
\exp \bigl((\alpha-1)&\dr{\alpha}{\mathcal{N}(\mu,\sigma^2)}{\mathcal{N}(\nu,\sigma^2)}\bigr) = \frac{1}{\sqrt{2\pi\sigma^2}}\int_\mathbb{R} \exp \left(-\alpha \frac{(x-\mu)^2}{2\sigma^2} - (1-\alpha) \frac{(x - \nu)^2}{2\sigma^2} \right) \mathrm{d}x\\
=& \frac{1}{\sqrt{2\pi\sigma^2}}\int_\mathbb{R} \exp \left(-\frac{(x - (\alpha\mu + (1-\alpha)\nu))^2 - (\alpha\mu + (1-\alpha)\nu)^2 + \alpha \mu^2 + (1-\alpha)\nu^2}{2\sigma^2} \right) \mathrm{d}x\\
=& \ex{x \sim \mathcal{N}(\alpha\mu + (1-\alpha)\nu,\sigma^2)}{ \exp \left( -\frac{-(\alpha\mu + (1-\alpha)\nu)^2 + \alpha \mu^2 + (1-\alpha)\nu^2}{2\sigma^2}\right) }\\
=& \exp \left(\frac{\alpha(\alpha-1)(\mu - \nu)^2}{2\sigma^2}\right).
\end{align*}
\end{proof}

For the multivariate Gaussian mechanism, Lemma \ref{lem:NormalDivergence} generalises to the following.

\begin{lem} \label{lem:NormalDivergenceMulti}
Let $\mu, \nu \in \mathbb{R}^d$, $\sigma \in \mathbb{R}$, and $\alpha \in [1,\infty)$. Then $$\dr{\alpha}{\mathcal{N}(\mu,\sigma^2 I_d)}{\mathcal{N}(\nu,\sigma^2 I_d)} = \frac{\alpha\|\mu - \nu\|_2^2}{2\sigma^2}$$
\end{lem}
Thus, if $M : \mathcal{X}^n \to \mathbb{R}^d$ is the mechanism that, on input $x$, releases a sample from $\mathcal{N}(q(x), \sigma^2 I_d)$ for some function $q : \mathcal{X}^n \to \mathbb{R}^d$, then $M$ satisfies $\lin$-\IIP{} for \begin{equation}\lin = \frac{1}{2\sigma^2} \sup_{x,x' \in \mathcal{X}^n \atop \text{differing in one entry}} \|q(x)-q(x')\|_2^2 .\end{equation}

\section{Relation to Differential Privacy} \label{sec:DP}

We now discuss the relationship between \IIP{} and the traditional definitions of pure and approximate differential privacy. There is a close relationship between the notions, but not an exact characterization.

For completeness, we state the definition of differential privacy:

\begin{defn}[Differential Privacy (DP) \cite{DworkMNS06,DworkKMMN06}]
A randomized mechanism $M : \mathcal{X}^n \to \mathcal{Y}$ satisfies $(\varepsilon,\delta)$-differential privacy if, for all $x,x' \in \mathcal{X}$ differing in a single entry, we have $$\pr{}{M(x) \in S} \leq e^\varepsilon \pr{}{M(x') \in S} + \delta$$ for all (measurable) $S \subset \mathcal{Y}$. Further define $\varepsilon$-differential privacy to be $(\varepsilon,0)$-differential privacy.
\end{defn}

\subsection{Pure DP versus \IIP{}}

Pure differential privacy is exactly characterized by $(\con, 0)$-\IIP{}:
\begin{lem}
A mechanism $M : \mathcal{X}^n \to \mathcal{Y}$ satisfies $\varepsilon$-DP if and only if it satisfies $(\varepsilon,0)$-\IIP{}.
\end{lem}
\begin{proof}
Let $x,x' \in \mathcal{X}^n$ be neighbouring. Suppose $M$ satisfies $\varepsilon$-DP. Then $\dr{\infty}{M(x)}{M(x')} \leq \varepsilon$. By monotonicity, $$\dr{\alpha}{M(x)}{M(x')} \leq \dr{\infty}{M(x)}{M(x')} \leq \varepsilon = \varepsilon + 0 \cdot \alpha$$ for all $\alpha$. So $M$ satisfies $(\varepsilon,0)$-\IIP{}.
Conversely, suppose $M$ satisfies $(\varepsilon,0)$-\IIP{}. Then $$\dr{\infty}{M(x)}{M(x')} = \lim_{\alpha \to \infty} \dr{\alpha}{M(x)}{M(x')} \leq \lim_{\alpha \to \infty} \varepsilon + 0 \cdot \alpha = \varepsilon.$$ Thus $M$ satisfies $\varepsilon$-DP.
\end{proof}

We now show that $\eps$-differential privacy implies $(\frac12\eps^2)$-\IIP{} (Proposition \ref{prop:PDPtoCDP-intro}).

\begin{prop} \label{prop:EpsSquared}
Let $P$ and $Q$ be probability distributions on $\Omega$ satisfying $\dr{\infty}{P}{Q} \leq \varepsilon$ and $\dr{\infty}{Q}{P} \leq \varepsilon$. Then $\dr{\alpha}{P}{Q} \leq \frac12 \varepsilon^2 \alpha$ for all $\alpha>1$.
\end{prop}

\begin{rem}
In particular, Proposition \ref{prop:EpsSquared} shows that the KL-divergence $\dr{1}{P}{Q} \le \frac{1}{2}\eps^2$. A bound on the KL-divergence between random variables in terms of their max-divergence is an important ingredient in the analysis of the advanced composition theorem \cite{DworkRV10}. Our bound sharpens (up to lower order terms) and, in our opinion, simplifies the previous bound of $\dr{1}{P}{Q} \le \frac{1}{2}\eps(e^{\eps}-1)$ proved by Dwork and Rothblum \cite{DworkR16}.
\end{rem}

\begin{proof}[Proof of Proposition \ref{prop:EpsSquared}.]
We may assume $\frac12 \varepsilon \alpha \leq 1$, as otherwise $\frac12 \varepsilon^2 \alpha > \varepsilon$, whence the result follows from monotonicity. We must show that $$e^{(\alpha-1)\dr{\alpha}{P}{Q}} = \ex{x \sim Q}{\left(\frac{P(x)}{Q(x)}\right)^\alpha} \leq e^{\frac12 \alpha(\alpha-1)\varepsilon^2}.$$
We know that $e^{-\varepsilon} \leq \frac{P(x)}{Q(x)} \leq e^{\varepsilon}$ for all $x$. Define a random function $A : \Omega \to \{e^{-\varepsilon},e^\varepsilon\}$ by $\ex{A}{A(x)}={\frac{P(x)}{Q(x)}}$ for all $x$. By Jensen's inequality, $$\ex{x \sim Q}{\left(\frac{P(x)}{Q(x)}\right)^{\alpha}} = \ex{x \sim Q}{\left(\ex{A}{A(x)}\right)^{\alpha}} \leq \ex{x \sim Q}{\ex{A}{A(x)^\alpha}} = \ex{A}{A^\alpha},$$ where $A$ denotes $A(x)$ for a random $x \sim Q$. We also have $\ex{A}{A} = \ex{x \sim Q}{\frac{P(x)}{Q(x)}}=1$. From this equation, we can conclude that $$\pr{A}{A=e^{-\varepsilon}} = \frac{e^\varepsilon-1}{e^\varepsilon-e^{-\varepsilon}} \qquad \text{and} \qquad \pr{A}{A=e^{\varepsilon}} = \frac{1-e^{-\varepsilon}}{e^\varepsilon-e^{-\varepsilon}}.$$ 
Thus 
\begin{align*}e^{(\alpha-1)\dr{\alpha}{P}{Q}} \leq& \ex{A}{A^\alpha}\\
=& \frac{e^\varepsilon-1}{e^\varepsilon-e^{-\varepsilon}} \cdot e^{-\alpha\varepsilon} + \frac{1-e^{-\varepsilon}}{e^\varepsilon-e^{-\varepsilon}} \cdot e^{\alpha\varepsilon}\\
=&\frac{(e^{\alpha\varepsilon}-e^{-\alpha\varepsilon}) - (e^{(\alpha-1)\varepsilon}-e^{-(\alpha-1)\varepsilon})}{e^\varepsilon-e^{-\varepsilon}}\\
=& \frac{\sinh(\alpha\varepsilon)-\sinh((\alpha-1)\varepsilon)}{\sinh(\varepsilon)}.
\end{align*}
The result now follows from the following inequality, which is proved in Lemma \ref{lem:HyperTrigIneq}. $$0 \leq y < x \leq 2 \implies \frac{\sinh(x)-\sinh(y)}{\sinh(x-y)} \leq e^{\frac12 xy}.$$
\end{proof}

\subsection{Approximate DP versus \IIP{}}

The statements in this section show that, up to some loss in parameters, \IIP{} is equivalent to a family of $(\eps, \delta)$-DP guarantees for all $\delta > 0$.

\begin{lem} \label{lem:CDPtoDP}
Let $M : \mathcal{X}^n \to \mathcal{Y}$ satisfy $(\con,\lin)$-\IIP{}. Then $M$ satisfies $(\varepsilon,\delta)$-DP for all $\delta>0$ and $$\varepsilon = \con + \lin + \sqrt{4\lin \log(1/\delta)}.$$
\end{lem}

Thus to achieve a given $(\varepsilon,\delta)$-DP guarantee it suffices to satisfy $(\con,\lin)$-\IIP{} with $$\lin = \left( \sqrt{\varepsilon-\con+\log(1/\delta)}-\sqrt{\log(1/\delta)}\right)^2 \approx \frac{(\varepsilon-\con)^2}{4\log(1/\delta)}.$$

\begin{proof}
Let $x,x' \in \mathcal{X}^n$ be neighbouring. Define $f(y) = \log(\pr{}{M(x)=y}/\pr{}{M(x')=y})$. Let $Y\sim M(x)$ and $Z=f(Y)$. That is, $Z=\privloss{M(x)}{M(x')}$ is the privacy loss random variable.
Fix $\alpha\in(1,\infty)$ to be chosen later.
Then $$\ex{}{e^{(\alpha-1)Z}}=\ex{Y \sim M(x)}{\left(\frac{\pr{}{M(x)=Y}}{\pr{}{M(x')=Y}}\right)^{\alpha-1}}=e^{(\alpha-1)\dr{\alpha}{M(x)}{M(x')}} \leq e^{(\alpha-1)(\con+\lin\alpha)}.$$
By Markov's inequality $$\pr{}{Z>\varepsilon} = \pr{}{e^{(\alpha-1) Z} > e^{(\alpha-1)\varepsilon}} \leq \frac{\ex{}{e^{(\alpha-1) Z}}}{e^{(\alpha-1)\varepsilon}} \leq e^{(\alpha-1)(\con+\lin\alpha-\varepsilon)}.$$ Choosing $\alpha = (\varepsilon - \con + \lin)/2\lin > 1$ gives $$\pr{}{Z>\varepsilon} \leq e^{-(\varepsilon - \con - \lin)^2/4\lin} \leq \delta.$$
Now, for any measurable $S \subset \mathcal{Y}$, 
\begin{align*}
\pr{}{M(x) \in S} =& \pr{}{Y \in S}\\
\leq& \pr{}{Y \in S \wedge Z \leq \varepsilon} +  \pr{}{Z > \varepsilon} \\
\leq& \pr{}{Y \in S \wedge Z \leq \varepsilon} + \delta\\
=& \int_{\mathcal{Y}} \pr{}{M(x)=y} \cdot \mathbb{I}(y \in S) \cdot \mathbb{I}(f(y) \leq \varepsilon) \ \mathrm{d}y + \delta\\
\leq& \int_{\mathcal{Y}} e^\varepsilon \pr{}{M(x')=y} \cdot \mathbb{I}(y \in S) \ \mathrm{d}y + \delta\\
=& e^\varepsilon \pr{}{M(x') \in S} + \delta.
\end{align*}
\end{proof}

Lemma \ref{lem:CDPtoDP} is not tight. In particular, we have the following refinement of Lemma \ref{lem:CDPtoDP}, the proof of which is deferred to the appendix (Lemma \ref{lem:RenyiToED}).

\begin{lem} \label{lem:CDPtoDP2}
Let $M : \mathcal{X}^n \to \mathcal{Y}$ satisfy $(\con,\lin)$-\IIP{}. Then $M$ satisfies $(\varepsilon,\delta)$-DP for all $\delta>0$ and $$\varepsilon = \con+\lin+\sqrt{4\lin\cdot\log(\sqrt{\pi \cdot \lin} / \delta)}.$$
Alternatively $M$ satisfies $(\varepsilon,\delta)$-DP for all $\varepsilon\geq\con+\lin$ and $$\delta = e^{-(\varepsilon-\con-\lin)^2/4\lin} \cdot \min \left\{ \begin{array}{l} \sqrt{\pi \cdot \lin} \\ \frac{1}{1+(\varepsilon-\con-\lin)/2\lin} \\ \frac{2}{1+\frac{\varepsilon-\con-\lin}{2\lin} + \sqrt{\left(1+\frac{\varepsilon-\con-\lin}{2\lin}\right)^2+\frac{4}{\pi \lin}}} \end{array} \right..$$
\end{lem}
Note that the last of three options in the minimum dominates the first two options. We have included the first two options as they are simpler.

Now we show a partial converse to Lemma \ref{lem:CDPtoDP}, which is proved in the appendix (Lemma \ref{lem:DPtoCDP-app}).

\begin{lem} \label{lem:DPtoCDP}
Let $M : \mathcal{X}^n \to \mathcal{Y}$ satisfy $(\varepsilon,\delta)$-DP for all $\delta>0$ and \begin{equation} \varepsilon = \hat \con  +\sqrt{\hat \lin\log(1/\delta)}\label{eqn:DP-first}\end{equation} for some constants $\hat \con, \hat \lin \in [0,1]$. Then $M$ is $\left(\hat \con-\frac{1}{4} \hat \lin+5\sqrt[4]{\hat \lin}, \frac{1}{4} \hat \lin\right)$-\IIP{}.
\end{lem}
Thus \IIP{} and DP are equivalent up to a (potentially substantial) loss in parameters and the quantification over all $\delta$.

\section{Zero- versus Mean-Concentrated Differential Privacy} \label{sec:CDP-DR}

We begin by stating the definition of mean-concentrated differential privacy:

\begin{defn}[Mean-Concentrated Differential Privacy (\mcdp{}) \cite{DworkR16}]
A randomized mechanism $M : \mathcal{X}^n \to \mathcal{Y}$ satisfies $(\mu,\tau)$-mean-concentrated differential privacy if, for all $x,x' \in \mathcal{X}^n$ differing in one entry, and letting $Z = \privloss{M(x)}{M(x')}$, we have $$\ex{}{Z} \leq \mu$$ and $$ \ex{}{e^{\lambda\left(Z-\ex{}{Z}\right)}} \leq e^{\lambda^2 \cdot \tau^2 / 2} $$ for all $\lambda \in \mathbb{R}$.
\end{defn}

In contrast $(\con,\lin)$-\IIP{} requires that, for all $\alpha \in (1,\infty)$, $\ex{}{e^{(\alpha-1)Z}} \leq e^{(\alpha-1)(\con+\lin\alpha)}$, where $Z\sim\privloss{M(x)}{M(x')}$ is the privacy loss random variable. 
We now show that these definitions are equivalent up to a (potentially significant) loss in parameters.

\begin{lem} \label{lem:CDPtoIIP}
If $M : \mathcal{X}^n \to \mathcal{Y}$ satisfies $(\mu,\tau)$-\mcdp{}, then $M$ satisfies $(\mu-\tau^2/2,\tau^2/2)$-\IIP{}.
\end{lem}
\begin{proof}
For all $\alpha\in(1,\infty)$,
$$\ex{}{e^{(\alpha-1)Z}} = \ex{}{e^{(\alpha-1)(Z-\ex{}{Z})}} \cdot e^{(\alpha-1)\ex{}{Z}} \leq e^{(\alpha-1)^2 \tau^2 /2} \cdot e^{(\alpha-1)\mu} = e^{(\alpha-1)(\mu-\tau^2/2 + \tau^2/2 \cdot \alpha)}.$$
\end{proof}

\begin{lem} \label{lem:IIPtoCDP}
If $M : \mathcal{X}^n \to \mathcal{Y}$ satisfies $(\con,\lin)$-\IIP{}, then $M$ satisfies $(\con+\lin,O(\sqrt{\con+2\lin}))$-\mcdp{}.
\end{lem}

The proof of Lemma \ref{lem:IIPtoCDP} is deferred to the appendix.

Thus we can convert $(\mu,\tau)$-\mcdp{} into $(\mu-\tau^2/2,\tau^2/2)$-\IIP{} and then back to $(\mu, O(\sqrt{\mu+\tau^2/2}))$-\mcdp{}. This may result in a large loss in parameters, which is why, for example, pure DP can be characterised in terms of \IIP{}, but not in terms of \mcdp{}.

We view \IIP{} as a relaxation of \mcdp{}; \mcdp{} requires the privacy loss to be ``tightly'' concentrated about its mean and that the mean is close to the origin. The triangle inequality then implies that the privacy loss is ``weakly'' concentrated about the origin. (The difference between ``tightly'' and ``weakly'' accounts for the use of the triangle inequality.)
On the other hand, \IIP{} direcly requires that the privacy loss is weakly concentrated about the origin. That is to say, \IIP{} gives a subgaussian bound on the privacy loss that is centered at zero, whereas \mcdp{} gives a subgaussian bound that is centered at the mean and separately bounds the mean.

There may be some advantage to the stronger requirement of \mcdp{}, either in terms of what kind of privacy guarantee it affords, or how it can be used as an analytic tool. However, it seems that for most applications, we only need what \IIP{} provides.

\section{Group Privacy} \label{sec:GroupPrivacy}

In this section we show that \IIP{} provides privacy protections to small groups of individuals.

\begin{defn}[\IIP{} for Groups]
We say that a mechanism $M : \mathcal{X}^n \to \mathcal{Y}$ provides $(\con,\lin)$-\IIP{} for groups of size $k$ if, for every $x,x' \in \mathcal{X}^n$ differing in at most $k$ entries, we have $$\forall \alpha \in (1,\infty) \qquad \dr{\alpha}{M(x)}{M(x')} \leq \con + \lin \cdot \alpha.$$
\end{defn}
The usual definition of \IIP{} only applies to groups of size $1$. Here we show that it implies bounds for all group sizes. We begin with a technical lemma.

\begin{lem}[Triangle-like Inequality for R\'enyi Divergence]
Let $P$, $Q$, and $R$ be probability distributions. Then \begin{equation}\dr{\alpha}{P}{Q} \leq \frac{k\alpha}{k\alpha-1}\dr{ \frac{k\alpha-1}{k-1}}{P}{R} + \dr{k\alpha}{R}{Q}\label{eqn:DivTri}\end{equation} for all $k,\alpha \in (1,\infty)$.
\end{lem}
\begin{proof}
Let $p=\frac{k\alpha-1}{\alpha(k-1)}$ and $q=\frac{k\alpha-1}{\alpha-1}$. Then $\frac{1}{p}+\frac{1}{q}=\frac{\alpha(k-1) + (\alpha-1)}{k\alpha-1}=1$.
By H\"older's inequality,
\begin{align*}
e^{(\alpha-1)\dr{\alpha}{P}{Q}}=& \int_\Omega P(x)^\alpha Q(x)^{1-\alpha} \mathrm{d}x\\
=& \int_\Omega P(x)^\alpha R(x)^{-\alpha} \cdot R(x)^{\alpha-1} Q(x)^{1-\alpha} \cdot R(x) \mathrm{d}x\\
=& \ex{x \sim R}{ \left(\frac{P(x)}{R(x)}\right)^\alpha \cdot \left(\frac{R(x)}{Q(x)}\right)^{\alpha-1}}\\
\leq& \ex{x \sim R}{ \left(\frac{P(x)}{R(x)}\right)^{p\alpha}}^{1/p} \cdot \ex{x \sim R}{\left(\frac{R(x)}{Q(x)}\right)^{q(\alpha-1)}}^{1/q}\\
=& e^{(p\alpha-1)\dr{p\alpha}{P}{R} /p} \cdot e^{q(\alpha-1)\dr{q(\alpha-1)+1}{R}{Q} /q}.
\end{align*}
Taking logarithms and rearranging gives $$\dr{\alpha}{P}{Q} \leq \frac{p\alpha-1}{p(\alpha-1)}\dr{p\alpha}{P}{R} + \dr{q(\alpha-1)+1}{R}{Q}.$$
Now $p\alpha=\frac{k\alpha-1}{k-1}$, $q(\alpha-1)+1=k\alpha$, and $$\frac{~~\frac{p\alpha-1}{p(\alpha-1)}~~}{~~\frac{k\alpha}{k\alpha-1}~~} = \frac{p\alpha-1}{p\alpha} \cdot \frac{k\alpha-1}{k(\alpha-1)} = \frac{\frac{k\alpha-1}{k-1}-1}{~\frac{k\alpha-1}{k-1}~} \cdot \frac{k\alpha-1}{k(\alpha-1)} = \frac{k\alpha-1 -k + 1}{k\alpha-1} \cdot \frac{k\alpha-1}{k(\alpha-1)} =1,$$ as required.
\end{proof}

\begin{prop} \label{prop:GroupPrivacy}
If $M : \mathcal{X}^n \to \mathcal{Y}$ satisfies $(\con,\lin)$-\IIP{}, then $M$ gives $(\con \cdot k \sum_{i=1}^k \frac{1}{i}, \lin \cdot k^2)$-\IIP{} for groups of size $k$.
\end{prop}
Note that $$\sum_{i=1}^k \frac{1}{i} = 1+\int_1^k \frac{1}{\lceil x \rceil} \mathrm{d}x \leq 1+\int_1^k \frac{1}{ x } \mathrm{d}x = 1 + \log k.$$
Thus $(\con,\lin)$-\IIP{} implies $(\con\cdot O(k \log k), \lin \cdot k^2)$-\IIP{} for groups of size $k$.

The Gaussian mechanism shows that $k^2\lin$ is the optimal dependence on $\lin$. However, $O(k \log k)\con$ is not the optimal dependence on $\con$: $(\con,0)$-\IIP{} implies $(k\con,0)$-\IIP{} for groups of size $k$.
\begin{proof}
We show this by induction on $k$. The statement is clearly true for groups of size $1$. We now assume the statement holds for groups of size $k-1$ and will verify it for groups of size $k$.

Let $x,x' \in \mathcal{X}^n$ differ in $k$ entries. Let $\hat x \in \mathcal{X}^n$ be such that $x$ and $\hat x$ differ in $k-1$ entries and $x'$ and $\hat x$ differ in one entry.

Then, by the induction hypothesis, $$\dr{\alpha}{M(x)}{M(\hat x)} \leq \con \cdot (k-1) \sum_{i=1}^{k-1} \frac{1}{i} + \lin \cdot (k-1)^2 \cdot \alpha$$ and, by \IIP{}, $$\dr{\alpha}{M(\hat x)}{M(x')} \leq \con + \lin \cdot  \alpha$$ for all $\alpha\in(1,\infty)$.

By \eqref{eqn:DivTri}, for any $\alpha\in(1,\infty)$,
\begin{align*}
\dr{\alpha}{M(x)}{M(x')} \leq& \frac{k\alpha}{k\alpha-1}\dr{ \frac{k\alpha-1}{k-1}}{M(x)}{M(\hat x)} + \dr{k\alpha}{M(\hat x)}{M(x')}\\
\leq& \frac{k\alpha}{k\alpha-1}\left( \con \cdot (k-1) \sum_{i=1}^{k-1} \frac{1}{i} + \lin \cdot (k-1)^2 \cdot \frac{k\alpha-1}{k-1} \right) + \con + \lin \cdot k\alpha\\
=& \con \cdot \left( 1 + \frac{k\alpha}{k\alpha-1} (k-1) \sum_{i=1}^{k-1} \frac{1}{i} \right) + \lin \cdot \left( \frac{k\alpha}{k\alpha-1} (k-1)^2 \frac{k\alpha-1}{k-1}  + k\alpha \right)\\
=& \con \cdot \left( 1 + \frac{k\alpha}{k\alpha-1} (k-1) \sum_{i=1}^{k-1} \frac{1}{i} \right) + \lin \cdot k^2 \cdot \alpha\\
\leq& \con \cdot \left( 1 + \frac{k}{k-1} (k-1) \sum_{i=1}^{k-1} \frac{1}{i} \right) + \lin \cdot k^2 \cdot \alpha\\
=& \con \cdot k \sum_{i=1}^{k} \frac{1}{i}  + \lin \cdot k^2 \cdot \alpha,
\end{align*}
where the last inequality follows from the fact that $\frac{k\alpha}{k\alpha-1}$ is a decreasing function of $\alpha$ for $\alpha>1$. 
\end{proof}

\section{Lower Bounds} \label{sec:LowerBounds}

In this section we develop tools to prove lower bounds for \IIP{}. We will use group privacy to bound the mutual information between the input and the output of a mechanism satisfying \IIP{}. Thus, if we are able to construct a distribution on inputs such that any accurate mechanism must reveal a high amount of information about its input, we obtain a lower bound showing that no accurate mechanism satisfying \IIP{} can be accurate for this data distribution.

We begin with the simplest form of our mutual information bound, which is an analogue of the bound of \cite{McGregorMPRTV10} for pure differential privacy:

\begin{prop} \label{prop:MutualInformation}
Let $M : \mathcal{X}^n \to \mathcal{Y}$ satisfy $(\con,\lin)$-\IIP{}. Let $X$ be a random variable in $\mathcal{X}^n$. Then $$I( X ; M(X) ) \leq \con \cdot n (1+\log n) + \lin \cdot n^2,$$ where $I$ denotes mutual information (measured in nats, rather than bits).
\end{prop}
\begin{proof}
By Proposition \ref{prop:GroupPrivacy}, $M$ provides $(\con \cdot n \sum_{i=1}^n \frac{1}{i}, \lin \cdot n^2)$-\IIP{} for groups of size $n$. Thus $$\dr{1}{M(x)}{M(x')} \leq \con \cdot n \sum_{i=1}^n \frac{1}{i} + \lin \cdot n^2 \leq \con \cdot n (1+\log n) + \lin \cdot n^2$$ for all $x,x' \in \mathcal{X}^n$.
Since KL-divergence is convex,
\begin{align*}
I( X ; M(X) ) =& \ex{x \leftarrow X}{\dr{1}{M(x)}{M(X)}}\\
\leq& \ex{x \leftarrow X}{\ex{x' \leftarrow X}{\dr{1}{M(x)}{M(x')}}}\\
\leq& \ex{x \leftarrow X}{\ex{x' \leftarrow X}{\con \cdot n (1+\log n) + \lin \cdot n^2}}\\
=& \con \cdot n (1+\log n) + \lin \cdot n^2.
\end{align*}
\end{proof}

The reason this lower bound works is the strong group privacy guarantee --- even for groups of size $n$, we obtain nontrivial privacy guarantees.
While this is good for privacy it is bad for usefulness, as it implies that even information that is ``global'' (rather than specific to a individual or a small group) is protected. These lower bounds reinforce the connection between group privacy and lower bounds \cite{HardtT10,De12,SteinkeU15b}.

In contrast, $(\varepsilon,\delta)$-DP is not susceptible to such a lower bound because it gives a vacuous privacy guarantee for groups of size $k=O(\log(1/\delta)/\varepsilon)$. This helps explain the power of the propose-test-release paradigm.

Furthermore, we obtain even stronger mutual information bounds when the entries of the distribution are independent:

\begin{lem} \label{lem:MutualInformation}
Let $M : \mathcal{X}^{m} \to \mathcal{Y}$ satisfy $(\con,\lin)$-\IIP{}. Let $X$ be a random variable in $\mathcal{X}^m$ with independent entries.  Then $$I \left( X ; M(X) \right) \leq (\con+\lin) \cdot m ,$$ where $I$ denotes mutual information (measured in nats, rather than bits).
\end{lem}
\begin{proof}
First, by the chain rule for mutual information, $$I(X; M(X)) = \sum_{i \in [m]} I(X_i ; M(X) | X_{1 \cdots i-1}),$$
where
\begin{align*}
I(X_i ; M(X) | X_{1 \cdots i-1}) =& \ex{x \leftarrow X_{1 \cdots i-1}}{I(X_i | X_{1 \cdots i-1} = x; M(X) | X_{1 \cdots i-1} = x)}\\
=& \ex{x \leftarrow X_{1 \cdots i-1}}{I(X_i; M(x,X_{i \cdots m}))},
\intertext{by independence of the $X_i$s.}
\intertext{We can define mutual information in terms of KL-divergence:}
I(X_i; M(x,X_{i \cdots m})) 
=& \ex{y \leftarrow X_i}{\dr{1}{M(x,X_{i \cdots m}) | X_i=y}{M(x,X_{i \cdots m})}} \\
=& \ex{y \leftarrow X_i}{\dr{1}{M(x,y,X_{i+1 \cdots m})}{M(x,X_{i \cdots m})}}.
\end{align*}
By \IIP{}, we know that for all $x \in \mathcal{X}^{i-1}$, $y,y' \in \mathcal{X}$, and $z \in \mathcal{X}^{m-i}$, we have $$\dr{1}{M(x,y,z)}{M(x,y',z)} \leq \con + \lin.$$
Thus, by the convexity of KL-divergence, $$\dr{1}{M(x,y,X_{i+1 \cdots m})}{M(x,X_{i \cdots m})} \leq \con+\lin$$ for all $x$ and $y$. The result follows.
\end{proof}

More generally, we can combine dependent and independent entries as follows.

\begin{thm} \label{thm:MixLB}
Let $M : \mathcal{X}^n \to \mathcal{Y}$ satisfy $(\con,\lin)$-\IIP{}. Take $n=m \cdot \ell$. Let $X^1, \cdots, X^m$ be independent random variables on $\mathcal{X}^\ell$. Denote $X=(X^1, \cdots, X^{m}) \in \mathcal{X}^n$. Then $$I\left(X;M(X)\right) \leq m \cdot \left(  \con \cdot \ell (1+\log \ell) + \lin \cdot \ell^2 \right),$$ where $I$ denotes the mutual information (measured in nats, rather than bits). 
\end{thm}
\begin{proof}
By Proposition \ref{prop:GroupPrivacy}, $M$ provides $(\con \cdot \ell \sum_{i=1}^\ell \frac{1}{i}, \lin \cdot \ell^2)$-\IIP{} for groups of size $\ell$. 
Thus \begin{equation}\dr{1}{M(x_1, \cdots, x_i, \cdots, x_m)}{M(x_1, \cdots, x_i', \cdots, x_m)} \leq \con \cdot \ell \sum_{i=1}^\ell \frac{1}{i} + \lin \cdot \ell^2 \leq \con \cdot \ell (1+\log \ell) + \lin \cdot \ell^2 \label{eqn:KLgroup}\end{equation} for all $x_1, \cdots, x_m, x_i' \in \mathcal{X}^\ell$.

By the chain rule for mutual information, $$I(X; M(X)) = \sum_{i \in [m]} I(X_i ; M(X) | X_{1 \cdots i-1}),$$
where
\begin{align*}
I(X_i ; M(X) | X_{1 \cdots i-1}) =& \ex{x \leftarrow X_{1 \cdots i-1}}{I(X_i | X_{1 \cdots i-1} = x; M(X) | X_{1 \cdots i-1} = x)}\\
=& \ex{x \leftarrow X_{1 \cdots i-1}}{I(X_i; M(x,X_{i \cdots m}))},
\intertext{by independence of the $X_i$s.}
\intertext{We can define mutual information in terms of KL-divergence:}
I(X_i; M(x,X_{i \cdots m})) 
=& \ex{y \leftarrow X_i}{\dr{1}{M(x,X_{i \cdots m}) | X_i=y}{M(x,X_{i \cdots m})}} \\
=& \ex{y \leftarrow X_i}{\dr{1}{M(x,y,X_{i+1 \cdots m})}{M(x,X_{i \cdots m})}}.
\end{align*}
By \eqref{eqn:KLgroup} and the convexity of KL-divergence, $$\dr{1}{M(x,y,X_{i+1 \cdots m})}{M(x,X_{i \cdots m})} \leq \con \cdot \ell (1+\log \ell) + \lin \cdot \ell^2$$ for all $x$ and $y$. The result follows.
\end{proof}

\subsection{Example Applications of the Lower Bound}

We informally discuss a few applications of our information-based lower bounds to some simple and well-studied problems in differential privacy.

\paragraph{One-Way Marginals}
Consider $M : \mathcal{X}^n \to \mathcal{Y}$ where $\mathcal{X}=\{0,1\}^d$ and $\mathcal{Y}=[0,1]^d$. The goal of $M$ is to estimate the attribute means, or one-way marginals, of its input database $x$: $$M(x) \approx \overline{x} = \frac{1}{n} \sum_{i \in [n]} x_i.$$ 

It is known that this is possible subject to $\varepsilon$-DP if and only if $n=\Theta(d/\varepsilon)$ \cite{HardtT10,SteinkeU15b}. This is possible subject to $(\varepsilon,\delta)$-DP if and only if $n=\tilde{\Theta}(\sqrt{d\log(1/\delta)}/\varepsilon)$, assuming $\delta \ll 1/n$ \cite{BunUV14,SteinkeU15b}.

We now analyze what can be accomplished with \IIP{}. Adding independent noise drawn from $\mathcal{N}(0,d/2n^2\lin)$ to each of the $d$ coordinates of $\overline{x}$ satisfies $\lin$-\IIP{}. This gives accurate answers as long as $n \gg \sqrt{d/\lin}$.

For a lower bound, consider sampling $X_1 \in \{0,1\}^d$ uniformly at random. Set $X_i=X_1$ for all $i \in [n]$. By Proposition \ref{prop:MutualInformation}, $$I(X ; M(X) ) \leq n^2 \lin$$ for any $\lin$-\IIP{} $M : (\{0,1\}^d)^n \to [0,1]^d$. However, if $M$ is accurate, we can recover (most of) $X_1$ from $M(X)$, whence $I(X ; M(X) ) \geq \Omega(d)$.
 This yields a lower bound of $n \geq \Omega(\sqrt{d/\lin})$, which is tight up to constant factors.

\paragraph{Histograms (a.k.a.~Point Queries)}
Consider $M : \mathcal{X}^n \to \mathcal{Y}$, where $\mathcal{X}=[T]$ and $\mathcal{Y}=\mathbb{R}^T$. The goal of $M$ is to estimate the histogram of its input: $$ M(x)_t \approx h_t(x) = |\{i \in [n] : x_i = t\}|$$

For $\varepsilon$-DP it is possible to do this if and only if $n = \Theta(\log(T)/\varepsilon))$; the optimal algorithm is to independently sample $$M(x)_t \sim h_t(x) + \mathsf{Laplace}(2/\varepsilon).$$ However, for $(\varepsilon,\delta)$-DP, it is possible to attain sample complexity $n = O(\log(1/\delta)/\varepsilon)$ \cite[Theorem 3.13]{BeimelNS13, BunNS16}.
Interestingly, for \IIP{} we can show that $n=\Theta(\sqrt{\log(T)/\lin})$ is sufficient and necessary:

Sampling $$M(x)_t \sim h_t(x) + \mathcal{N}(0,1/\lin)$$ independently for $t \in [T]$ satisfies $\lin$-\IIP{}. Moreover, $$\pr{}{\max_{t \in [T]} \left| M(x)_t - h_t(x) \right| \geq \lambda} \leq T \cdot \pr{}{|\mathcal{N}(0,1/\lin)| > \lambda} \leq T \cdot e^{-\lambda^2\lin/2}.$$ In particular $\pr{}{\max_{t \in [T]} \left| M(x)_t - h_t(x) \right| \geq \sqrt{\log(T/\beta)/\lin}} \leq \beta$ for all $\beta>0$. Thus this algorithm is accurate if $n \gg \sqrt{\log(T)/\lin}$.

On the other hand, if we sample $X_1 \in [T]$ uniformly at random and set $X_i=X_1$ for all $i \in [n]$, then $I(X;M(X)) \geq \Omega(\log T)$ for any accurate $M$, as we can recover $X_1$ from $M(X)$ if $M$ is accurate. Proposition \ref{prop:MutualInformation} thus implies that $n \geq \Omega(\sqrt{\log(T)/\lin})$ is necessary to obtain accuracy.

This gives a strong separation between approximate DP and \IIP{}.

\paragraph{Randomized Response and the Exponential Mechanism}

Consider $M : \{\pm 1\}^n \to \{\pm 1\}^n$ where the goal is to maximize $\langle x , M(x) \rangle$ subject to $\lin$-\IIP{}.

One solution is randomized response \cite{Warner65}: Each output bit $i$ of $M$ is chosen independently with $$\pr{}{M(x)_i=x_i} = \frac{e^\varepsilon}{e^\varepsilon+1}.$$ This satisfies $\varepsilon$-DP and, hence, $\frac12 \varepsilon^2$-\IIP{}. And $\ex{}{\langle x, M(x) \rangle} = n (e^\varepsilon-1)/(e^\varepsilon+1) = \Theta(n \varepsilon)$.
Alternatively, we can independently choose the output bits $i$ according to $$M(x)_i = \mathrm{sign}\left( \mathcal{N}(x_i, \sqrt{2}/\lin) \right),$$ which satisfies $\lin$-\IIP{}.

Turning our attention to lower bounds: Let $X \in \{\pm 1\}^n$ be uniformly random. By Lemma \ref{lem:MutualInformation}, since the bits of $X$ are independent, we have $I(X; M(X)) \leq \lin \cdot n$ for any $\lin$-\IIP{} $M$. However, if $M$ is accurate, we can recover part of $X$ from $M(X)$ \cite{BunSU16}, whence $I(X;M(X)) \geq \Omega(n)$.

Randomized response is a special case of the exponential mechanism \cite{McSherryT07,BlumLR08}. Consequently this can be interpreted as a lower bound for search problems.

\paragraph{Lower Bounds with Accuracy}
The above examples can be easily discussed in terms of a more formal and quantitative definition of accuracy. In particular, we consider the histogram example again:
\begin{prop}
If $M : [T]^n \to \mathbb{R}^T$ satisfies $\lin$-\IIP{} and  $$\forall x \in [T]^n \qquad \ex{M}{\max_{t \in [T]} \big| M(x)_t - h_t(x) \big|} \leq \alpha n,$$ then $n \geq \Omega(\sqrt{\log(\alpha^2  T)/\rho \alpha^2})$.
\end{prop}
\begin{proof}
Let $m=1/10\alpha$ and $\ell=n/m$. For simplicity, assume that both $m$ and $n$ are integral.

Let $X_1, X_2, \cdots, X_m \in [T]^\ell$ be independent, where each $X_i$ is $\ell$ copies of a uniformly random element of $[T]$. By Theorem \ref{thm:MixLB}, \begin{equation}I(X;M(X)) \leq \lin \cdot m \cdot \ell^2 = 10 \lin \alpha n^2,\label{eqn:Iupper}\end{equation} where $X = (X_1, \cdots, X_m) \in \mathcal{X}^n$. 
However, \begin{align*} I(X;M(X)) \geq& I(f(X);g(M(X))) \\=& H(f(X))-H(f(X)|g(M(X))) \\=& H(X) - H(X|f(X))-H(f(X)|g(M(X)))\end{align*} for any functions $f$ and $g$, where $H$ is the entropy (in nats). In particular, we let $$f(x)= \{t \in T : \exists i \in [n] ~~ x_i = t\} \qquad \text{and} \qquad g(y) = \{t \in T : y_t \geq 5 \alpha n\}.$$
Clearly $H(X)=m \log T$. Furthermore, $H(X|f(X)) \leq m \log m$, since $X$ can be specified by naming $m$ elements of $f(X)$, which is a set of at most $m$ elements.

If \begin{equation}\max_{t \in [T]} \big| M(X)_t - h_t(X) \big| < 5\alpha, \label{eqn:AccEvent}\end{equation} then $g(M(X))$ contains exactly all the values in $X$ --- i.e. $f(X)=g(M(X))$. By Markov's inequality, \eqref{eqn:AccEvent} holds with probability at least $4/5$. 

Now we can upper bound $H(f(X)|g(M(X)))$ by giving a scheme for specifying $f(X)$ given $g(M(X))$. If \eqref{eqn:AccEvent} holds, we simply need one bit to say so. If \eqref{eqn:AccEvent} does not hold, we need one bit to say this and $m \log_2 T$ bits to describe $f(X)$. This gives $$H(f(X)|g(M(X))) \leq \log 2 + \pr{}{f(X) \ne g(M(X))} \cdot m \log T.$$

Combining these inequalities gives $$I(X;M(X)) \geq m \log T - m \log m - \log 2 - \frac{1}{5} m \log T \geq \frac{4}{5} m \log (T m^{-5/4}) - 1 \geq \Omega(\log(\alpha^{1.25} T)/\alpha).$$
Combining this with \eqref{eqn:Iupper} completes the proof.
\end{proof}

We remark that our lower bounds for \IIP{} can be converted to lower bounds for \mcdp{} using Lemma \ref{lem:CDPtoIIP}.

\section{Obtaining Pure DP Mechanisms from \IIP{}} \label{sec:CDPvsPDP}

We now establish limits on what more can be achieved with \IIP{} over pure differential privacy. In particular, we will prove that any mechanism satisfying \IIP{} can be converted into a mechanism satisfying pure DP with at most a quadratic blowup in sample complexity. Formally, we show the following theorem.

\begin{thm} \label{thm:CDPtoPDP}
Fix $n \in \mathbb{N}$, $n' \in \mathbb{N}$, $k \in \mathbb{N}$ $\alpha>0$, and $\varepsilon>0$. Let $q : \mathcal{X} \to \mathbb{R}^k$ and let $\|\cdot\|$ be a norm on $\mathbb{R}^k$. Assume $\max_{x \in \mathcal{X}} \|q(x)\| \leq 1$. 

Suppose there exists a $(\con,\lin)$-\IIP{} mechanism $M : \mathcal{X}^n \to \mathbb{R}^k$ such that for all $x \in \mathcal{X}^n$, $$\ex{M}{\|M(x)-q(x)\|} \leq {\alpha}.$$

Assume $\con \leq \alpha^2$, $\lin \leq \alpha^2$, and $$n' \geq \frac{4}{\varepsilon\alpha} \left( \lin \cdot n^2 + \con \cdot n \cdot (1 + \log n)  + 1 \right).$$

Then there exists a $(\varepsilon,0)$-differentially private $M' : \mathcal{X}^{n'} \to \mathbb{R}^k$ satisfying $$\ex{M'}{\|M'(x)-q(x)\|} \leq 10\alpha$$ and $$\pr{M'}{\|M'(x)-q(x)\| > 10\alpha + \frac{4}{\varepsilon n'} \log\left(\frac{1}{\beta}\right)} \leq \beta$$ for all $x \in \mathcal{X}^{n'}$ and $\beta>0$.
\end{thm}

Before discussing the proof of Theorem \ref{thm:CDPtoPDP}, we make some remarks about its statement:
\begin{itemize}
\item Unfortunately, the theorem only works for families of statistical queries $q : \mathcal{X} \to \mathbb{R}^k$. However, it works equally well for $\|\cdot\|_\infty$ and $\|\cdot\|_1$ error bounds.
\item If $\con=0$, we have $n' = O(n^2 \lin / \varepsilon \alpha)$. So, if $\lin$, $\varepsilon$, and $\alpha$ are all constants, we have $n'=O(n^2)$. This justifies our informal statement that we can convert any mechanism satisfying \IIP{} into one satisfying pure DP with a quadratic blowup in sample complexity.
\item Suppose $M : \mathcal{X}^n \to \mathbb{R}^k$ is the Gaussian mechanism scaled to satisfy $\lin$-\IIP{} and $\|\cdot\| = \|\cdot\|_1/{k}$. Then $$\alpha = \ex{}{\|M(x)-q(x)\|} = \Theta\left(\sqrt{\frac{k}{\lin n^2}}\right).$$ In particular, $n=\Theta(\sqrt{k/\lin\alpha^2})$.
The theorem then gives us a $\varepsilon$-DP $M' : \mathcal{X}^{n'} \to \mathbb{R}^k$ with $\ex{}{\|M'(x)-q(x)\|} \leq O(\alpha) $ for $$n' = \Theta\left(\frac{n^2 \lin }{ \varepsilon \alpha}\right) =  \Theta\left(\frac{k}{\alpha^3 \varepsilon}\right).$$
However, the Laplace mechanism achieves $\varepsilon$-DP and $\ex{}{\|M'(x)-q(x)\|} \leq \alpha $ with $n=\Theta(k/\alpha\varepsilon)$.

This example illustrates that the theorem is not tight in terms of $\alpha$; it loses a $1/\alpha^2$ factor here. However, the other parameters are tight. 

\item The requirement that $\con, \lin \leq \alpha^2$ is only used to show that \begin{equation}\max_{x \in \mathcal{X}^{n'}} \min_{\hat x \in \mathcal{X}^n} \|q(x)-q(\hat x)\| \leq 2\alpha \label{eqn:MaxMinSubSamp}\end{equation} using Lemma \ref{lem:SubSampAcc}. However, in many situations \eqref{eqn:MaxMinSubSamp} holds even when $\con,\lin \gg \alpha^2$.
For example, if $n \geq O(\log(k)/\alpha^2)$ or even $n \geq O(VC(q)/\alpha^2)$ then \eqref{eqn:MaxMinSubSamp} is automatically satisfied.

The technical condition \eqref{eqn:MaxMinSubSamp} is needed to relate the part of the proof with inputs of size $n$ to the part with inputs of size $n'$. 

Thus we can restate Theorem \ref{thm:CDPtoPDP} with the condition $\con, \lin \leq \alpha^2$ replaced by \eqref{eqn:MaxMinSubSamp}. This would be more general, but also more mysterious.
\end{itemize}
Alas, the proof of Theorem \ref{thm:CDPtoPDP} is not constructive. Rather than directly constructing a mechanism satisfying pure DP from any mechanism satisfying \IIP{}, we show the contrapositive statement: any lower bound for pure DP can be converted into a lower bound for \IIP{}. Pure DP is characterized by so-called packing lower bounds and the exponential mechanism.

We begin by giving a technical lemma showing that for any output space and any desired accuracy we have a ``packing'' and a ``net:''

\begin{lem} \label{lem:PackNet}
Let $(\mathcal{Y},d)$  be a metric space. Fix $\alpha>0$. Then there exists a countable $T \subset \mathcal{Y}$ such that both of the following hold.
\begin{itemize}
\item (Net:) Either $T$ is infinite or for all $y' \in \mathcal{Y}$ there exists $y \in T$ with $d(y,y') \leq \alpha$.
\item (Packing:) For all $y,y' \in T$, if $y \ne y'$, then $d(y,y')>\alpha$.
\end{itemize}
\end{lem}
\begin{proof}
Consider the following procedure for producing $T$.
\begin{itemize}
\item Initialize $A\leftarrow\mathcal{Y}$ and $T\leftarrow\emptyset$.
\item Repeat:
\begin{itemize}
\item If $A=\emptyset$, terminate.
\item Pick some $y \in A$.
\item Update $T \leftarrow T \cup \{y\}$.
\item Update $A\leftarrow\{y' \in A : d(y',y) > \alpha\}$.
\end{itemize}
\end{itemize}

This procedure either terminates giving a finite $T$ or runs forever enumerating a countably infinite $T$.

(Net:) If $T$ is infinite, we immediately can dispense the first condition, so suppose the procedure terminates and $T$ is finite.  Fix $y' \in \mathcal{Y}$.  Since the procedure terminates, $A=\emptyset$ at the end, which means $y'$ was removed from $A$ at some point. This means some $y \in T$ was added such that $d(y',y) \leq \alpha$, as required.

(Packing:) Fix $y \ne y' \in T$. We assume, without loss of generality, that $y$ was added to $T$ before $y'$. This means $y'$ was not removed from $A$ when $y$ was added to $T$. In particular, this means $d(y',y)>\alpha$.
\end{proof}

It is well-known that a net yields a pure DP algorithm:

\begin{lem}[Exponential Mechanism \cite{McSherryT07,BlumLR08}] \label{lem:ExpMech}
Let $\ell : \mathcal{X}^n \times T \to \mathbb{R}$ satisfy $|\ell(x,y) - \ell(x',y)| \leq \Delta$ for all 
$x,x' \in \mathcal{X}^n$ differing in one entry and all $y \in T$. Then, for all $\varepsilon>0$, there exists an $\varepsilon$-differentially private $M : \mathcal{X}^n \to T$ such that $$\pr{M}{\ell(x,M(x)) \leq \min_{y \in T} \ell(x,y) + \frac{2\Delta}{\varepsilon}\log\left(\frac{|T|}{\beta}\right)} \geq 1-\beta $$ and $$\ex{M}{\ell(x,M(x))} \leq \min_{y \in T} \ell(x,y) + \frac{2\Delta}{\varepsilon} \log |T|$$ for all $x \in \mathcal{X}^n$ and $\beta>0$.
\end{lem}
\begin{proof} The mechanism is defined by
$$\pr{M}{M(x)=y} = \frac{e^{-\ell(x,y) \varepsilon/2\Delta}}{\sum_{y' \in T} e^{-\ell(x,y) \varepsilon/2\Delta}}. $$
The analysis can be found in \cite[Theorems 3.10 and 3.11]{DworkR14} and \cite[Lemma 7.1]{BassilyNSSSU16}.
\end{proof}

We also show that a packing yields a lower bound for \IIP{}:

\begin{lem} \label{lem:PackLB}
Let $(\mathcal{Y},d)$ be a metric space and $q : \mathcal{X}^n \to \mathcal{Y}$ a function. Let $M : \mathcal{X}^n \to \mathcal{Y}$ be a $(\con,\lin)$-\IIP{} mechanism satisfying $$\pr{M}{d(M(x),q(x)) > \alpha/2} \leq \beta$$ for all $x \in \mathcal{X}^n$. 
Let $T \subset \mathcal{Y}$ be such that $d(y,y')>\alpha$, for all $y, y' \in T$ with $y \ne y'$. Assume that for all $y \in T$ there exists $x \in \mathcal{X}^n$ with $q(x)=y$.
Then $$(1-\beta)\log|T| - \log 2 \leq \con \cdot n (1+\log n) + \lin \cdot n^2.$$
\end{lem}
In particular, if $\con=0$, we have $$n \geq \sqrt{\frac{(1-\beta)\log|T|-\log 2}{\lin}}= \Omega(\sqrt{\log|T|/\lin}).$$
\begin{proof}
Let $q^{-1} : T \to \mathcal{X}^n$ be a function such that $q(q^{-1}(y))=y$ for all $y \in T$. Define $f : \mathcal{Y} \to T$ by $$f(y) = \underset{y' \in T}{\mathrm{argmin}}~ d(y,y')$$ (breaking ties arbitrarily). Then $$\pr{M}{f(M(q^{-1}(y)))=y} \geq 1-\beta$$ for all $y\in T$, as $\pr{M}{d(M(q^{-1}(y)),y)>\alpha/2} \leq \beta$ and $d(y',y)>\alpha$ for all $y' \in T \setminus \{y\}$. 

Let $Y$ be a uniformly random element of $T$ and let $X=q^{-1}(Y)$. By  the data processing inequality and Proposition \ref{prop:MutualInformation}, $$I(Y; f(M(q^{-1}(Y))))=I(q(X);f(M(X))) \leq I(X;M(X)) \leq \con \cdot n (1+\log n) + \lin \cdot n^2.$$
However, $\pr{}{f(M(q^{-1}(Y)))=Y}\geq 1-\beta$. Denote $Z=f(M(q^{-1}(Y)))$ and let $E$ be the indicator of the event that $Z=Y$. We have
$$I(Y;Z) = H(Y) - H(Y|Z)= H(Y) - H(Y,E|Z)= H(Y) - H(Y|E,Z) - H(E|Z).$$
Clearly $H(Y)=\log|T|$ and $H(E|Z) \leq H(E) \leq \log 2$. Moreover,
\begin{align*}
H(Y|E,Z) =& \ex{e \leftarrow E}{H(Y|Z,E=e)}\\
=& \pr{}{Y=Z} \cdot 0 + \pr{}{Y \ne Z} \cdot H(Y|Z,Y \ne Z)\\
\leq& \beta \cdot H(Y).
\end{align*}
Thus $$I(Y;Z) \geq \log|T| - \log 2 - \beta \log|T|.$$
The result now follows by combining inequalities.
\end{proof}

We need one final technical lemma:

\begin{lem} \label{lem:SubSampAcc}
Let $q : \mathcal{X} \to \mathbb{R}^k$ satisfy $\max_{x \in \mathcal{X}} \|q(x)\| \leq 1$, where $\|\cdot\|$ is some norm. Let $M : \mathcal{X}^n \to \mathbb{R}^k$ satisfy $(\con,\lin)$-\IIP{} and $$\ex{M}{\|M(x)-q(x)\|} \leq \alpha $$ for all $x \in \mathcal{X}^n$.
For all $n'$, $$\max_{x \in \mathcal{X}^{n'}} \min_{\hat x \in \mathcal{X}^n} \|q(\hat x) - q(x)\| \leq 2\alpha + \sqrt{2(\con+\lin)}.$$
\end{lem}
The proof of Lemma \ref{lem:SubSampAcc} is deferred to the appendix.

Now we can combine Lemmas \ref{lem:PackNet}, \ref{lem:ExpMech}, \ref{lem:PackLB}, and \ref{lem:SubSampAcc} to prove Theorem \ref{thm:CDPtoPDP}:
\begin{proof}[Proof of Theorem \ref{thm:CDPtoPDP}]
Apply Lemma \ref{lem:PackNet} with $\mathcal{Y} = \left\{q(x) = \frac{1}{n} \sum_{i \in [n]} q(x_i) : x \in \mathcal{X}^n \right\} \subset \mathbb{R}^k$ and $d$ being the metric induced by the norm to obtain $T \subset \mathcal{Y}$:
\begin{itemize}
\item (Net:) Either $T$ is infinite or for all $y' \in \{q(x) : x \in \mathcal{X}^n \} \subset \mathbb{R}^k$ there exists $y \in T$ with $\|y-y'\| \leq 4\alpha$.
\item (Packing:) For all $y,y' \in T$, if $y \ne y'$, then $\|y-y'\|>4\alpha$.
\end{itemize}

By Markov's inequality $$\pr{M}{\|M(x)-q(x)\| > 2\alpha} \leq \frac12.$$
Thus, by Lemma \ref{lem:PackLB}, $$\frac12 \log|T| - \log 2 \leq \con \cdot n (1+\log n) + \lin \cdot n^2.$$
This gives an upper bound on $|T|$. In particular, $T$ must be finite.

Let $M' : \mathcal{X}^{n'} \to \mathbb{R}^k$ be the exponential mechanism (Lemma \ref{lem:ExpMech}) instantiated with $T$ and $\ell(x,y) = \|y-q(x)\|$. We have $$\pr{M}{\|M(x)-q(x)\| \leq \min_{y \in T} \| y-q(x)\| +  \frac{4}{\varepsilon n'}\log\left(\frac{|T|}{\beta}\right)} \geq 1-\beta $$ and $$\ex{M}{\|M(x)-q(x)\|} \leq \min_{y \in T} \| y-q(x)\| +   \frac{4}{\varepsilon n'} \log |T|$$ for all $x \in \mathcal{X}^{n'}$.
For $x \in \mathcal{X}^{n'}$, by the Net property and Lemma \ref{lem:SubSampAcc},
\begin{align*}
\min_{y \in T} \| y-q(x)\| \leq& \min_{y \in T} \min_{y' \in \mathcal{Y}} \| y - y' \| + \|y' - q(x)\|\\
=&  \min_{y' \in \mathcal{Y}} \left(\left(\min_{y \in T}\| y - y' \|\right) + \|y' - q(x)\|\right)\\
\leq&  \min_{y' \in \mathcal{Y}} \left(4\alpha + \|y' - q(x)\|\right)\\
=& \min_{\hat x \in \mathcal{X}^n} \left(4\alpha + \|q(\hat x) - q(x)\|\right)\\
\leq& 4\alpha + 2\alpha + \sqrt{2(\con+\lin)}.
\end{align*}
Furthermore,
$$\frac{4}{\varepsilon n'} \log |T| \leq \frac{8}{\varepsilon n'} \left( \con \cdot n (1+\log n) + \lin \cdot n^2 + \log 2\right)\leq 2 \alpha.$$
The theorem now follows by combining inequalities.
\end{proof}

\section{Approximate \IIP{}} \label{sec:ApproxCDP}

In the spirit of approximate DP, we propose a relaxation of \IIP{}:

\begin{defn}[Approximate \IIP{}] \label{defn:ApproxRDP}
A randomised mechanism $M : \mathcal{X}^n \to \mathcal{Y}$ is $\delta$-approximately $(\con,\lin)$-\IIP{} if, for all $x,x' \in \mathcal{X}^n$ differing on a single entry, there exist events $E=E(M(x))$ and $E'=E'(M(x'))$ such that, for all $\alpha \in (1,\infty)$, $$ \dr{\alpha}{M(x)|_E}{M(x')|_{E'}} \leq \con + \lin \cdot \alpha ~~~~~ \text{and} ~~~~~ \dr{\alpha}{M(x')|_{E'}}{M(x)|_{E}} \leq \con + \lin \cdot \alpha$$ and $\pr{M(x)}{E} \geq 1- \delta$ and $\pr{M(x')}{E'} \geq 1-\delta$.
\end{defn}

Clearly $0$-approximate \IIP{} is simply \IIP{}. Hence we have a generalization of \IIP{}. As we will show later in this section, $\delta$-approximate $(\varepsilon,0)$-\IIP{} is equivalent to $(\varepsilon,\delta)$-DP. Thus we have also generalized approximate DP. Hence, this definition unifies both relaxations of pure DP. 

Approximate \IIP{} is a three-parameter definition which allows us to capture many different aspects of differential privacy. However, three parameters is quite overwhelming. We believe that use of the one-parameter $\lin$-\IIP{} (or the two-parameter $\delta$-approximate $\lin$-\IIP{} if necessary) is sufficient for most purposes.

It is easy to verify that the definition of approximate \IIP{} satisfies the following basic properties.

\begin{lem}[Composition \& Postprocessing] \label{lem:CompPost-ApproxGDP}
Let $M : \mathcal{X}^n \to \mathcal{Y}$ and $M' : \mathcal{X}^n \times \mathcal{Y} \to \mathcal{Z}$ be randomized algorithms. Suppose $M$ satisfies $\delta$-approximate $(\con,\lin)$-\IIP{} and, for all $y \in \mathcal{Y}$, $M'(\cdot,y) : \mathcal{X}^n \to \mathcal{Z}$ satisfies $\delta'$-approximate $(\con',\lin')$-\IIP{}. Define $M'' : \mathcal{X}^n \to \mathcal{Z}$ by $M''(x)=M'(x,M(x))$. Then $M''$ satisfies $(\delta+\delta'-\delta \cdot \delta')$-approximate $(\con+\con',\lin+\lin')$-\IIP{}.  
\end{lem}

\begin{lem}[Tradeoff]
Suppose $M : \mathcal{X}^n \to \mathcal{Y}$ satisfies $\delta$-approximate $(\con,0)$-\IIP{}. Then $M$ satisfies $\delta$-approximate $\con$-\IIP{} and $\delta$-approximate $\frac12 \con^2$-\IIP{}.
\end{lem}

However, the strong group privacy guarantees of Section \ref{sec:GroupPrivacy} no longer apply to approximate \IIP{} and, hence, the strong lower bounds of Section \ref{sec:LowerBounds} also no longer hold. Circumventing these lower bounds is part of the motivation for considering approximate \IIP{}. 
However, approximate \IIP{} is not necessarily the only way to relax \IIP{} that circumvents our lower bounds:

Proving the group privacy bound requires ``inflating'' the parameter $\alpha$: Suppose $M : \mathcal{X}^n \to \mathcal{Y}$ satisfies $\rho$-\IIP{} and $x,x' \in \mathcal{X}^n$ differ on $k$ entries. To prove $\dr{\alpha}{M(x)}{M(x')} \leq k^2\rho \alpha$, the proof of Proposition \ref{prop:GroupPrivacy} requires a bound on $\dr{k\alpha}{M(x'')}{M(x''')}$ for $x'',x''' \in \mathcal{X}^n$ differing on a single entry.

Consider relaxing the definition of \IIP{} to only require the bound \eqref{eqn:IIP-Renyi} or \eqref{eqn:GDP-MGF} to hold when $\alpha \leq m$:

\begin{defn}[Bounded \IIP{}] \label{def:BoundedCDP}
We say that $M : \mathcal{X}^n \to \mathcal{Y}$ satisfies $m$-bounded $(\con,\lin)$-\IIP{} if, for all $x,x' \in \mathcal{X}^n$ differing in only one entry and all $\alpha \in (1,m)$, $\dr{\alpha}{M(x)}{M(x')} \leq \con + \lin \cdot \alpha$.
\end{defn}

This relaxed definition may also be able to circumvent the group privacy-based lower bounds, as our group privacy proof would no longer work for groups of size larger than $m$. We do not know what group privacy guarantees Definition \ref{def:BoundedCDP} provides for groups of size $k \gg m$. This relaxed definition may be worth exploring, but is beyond the scope of our work.

\subsection{Approximate DP Implies Approximate \IIP{}}

We can convert approximate DP to approximate \IIP{} using the following lemma.

First we define a approximate DP version of the randomized response mechanism:

\begin{defn}
For $\varepsilon \geq 0$ and $\delta \in [0,1]$, define $\tilde{M}_{\varepsilon,\delta}: \{0,1\} \to \{0,1\} \times  \{\bot,\top\}$ by 
\begin{align*}
\pr{}{\tilde{M}_{\varepsilon,\delta}(b)=(b,\top)} =& \delta, &
\pr{}{\tilde{M}_{\varepsilon,\delta}(b)=(1-b,\top)} =& 0,\\
\pr{}{\tilde{M}_{\varepsilon,\delta}(b)=(b,\bot)} =& (1-\delta) \frac{e^\varepsilon}{1+e^\varepsilon}, &
\pr{}{\tilde{M}_{\varepsilon,\delta}(b)=(1-b,\bot)} =& (1-\delta) \frac{1}{1+e^\varepsilon}
\end{align*}
for both $b \in \{0,1\}$.
\end{defn}

The above mechanism is ``complete'' for approximate DP:

\begin{lem}[{\cite{KairouzOV15}, \cite[Lemma 3.2]{MurtaghV16}}] \label{lem:KOV-MV}
For every $(\varepsilon,\delta)$-DP $M : \mathcal{X}^n \to \mathcal{Y}$ and all $x_0,x_1 \in \mathcal{X}^n$ differing in one entry, there exists a randomized $T : \{0,1\} \times \{\bot,\top\} \to \mathcal{Y}$ such that $T(\tilde{M}_{\varepsilon,\delta}(b))$ has the same distribution as $M(x_b)$ for both $b \in \{0,1\}$.
\end{lem}

\begin{cor} \label{cor:DPtoApproxGDP}
If $M : \mathcal{X}^n \to \mathcal{Y}$ satisfies $(\varepsilon,\delta)$-DP, then $M$ satisfies $\delta$-approximate $(\varepsilon,0)$-\IIP{}, which, in turn, implies $\delta$-approximate $(0,\frac12 \varepsilon^2)$-\IIP{}.
\end{cor}
\begin{proof}
Fix neighbouring $x_0,x_1 \in \mathcal{X}^n$. Let $T : \{0,1\} \times \{\bot,\top\} \to \mathcal{Y}$ be as in Lemma \ref{lem:KOV-MV}.

Now we can write $M(x_b) = T(\tilde{M}_{\varepsilon,\delta}(b))$ for $b \in \{0,1\}$. Define events $E_0$ and $E_1$ by $$ E_b \equiv \left[\tilde{M}_{\varepsilon,\delta}(b) \in \{0,1\} \times \{\bot\} \right]. $$
By definition, for both $b \in \{0,1\}$, $\pr{\tilde{M}_{\varepsilon,\delta}(b)}{E_b} = 1-\delta$ and $$M(x_b)|_{E_b} = T\left(\tilde{M}_{\varepsilon,\delta}(b)|_{\tilde{M}_{\varepsilon,\delta}(b) \in \{0,1\} \times \{\bot\}}\right) = T(\tilde{M}_{\varepsilon,0}(b)).$$

We have $\dr{\infty}{\tilde{M}_{\varepsilon,0}(b)}{\tilde{M}_{\varepsilon,0}(1-b)} \leq \varepsilon$ for both $b \in \{0,1\}$. By postprocessing and monotonicity, this implies $$\dr{\alpha}{M(x_b)|_{E_b} }{M(x_{1-b})|_{E_{1-b}} } \leq \varepsilon$$ for both $b \in \{0,1\}$ and all $\alpha \in (1,\infty)$. Thus we have satisfied the definition of $\delta$-approximate $(\varepsilon,0)$-\IIP{}. 

Applying Proposition \ref{prop:EpsSquared} shows that this also implies $\delta$-approximate $(0,\frac12 \varepsilon^2)$-\IIP{}.
\end{proof}

\subsection{Approximate \IIP{} Implies Approximate DP}

\begin{lem} \label{lem:ApproxGDPtoDP}
Suppose $M : \mathcal{X}^n \to \mathcal{Y}$ satisfies $\delta$-approximate $(\con,\lin)$-\IIP{}. If $\lin=0$, then $M$ satisfies $(\con,\delta)$-DP. In general, $M$ satisfies $(\varepsilon,\delta+(1-\delta)\delta')$-DP for all $\varepsilon\geq\con+\lin$, where $$\delta' = e^{-(\varepsilon-\con-\lin)^2/4\lin} \cdot \min \left\{ \begin{array}{l} 1 \\ \sqrt{\pi \cdot \lin} \\ \frac{1}{1+(\varepsilon-\con-\lin)/2\lin} \\ \frac{2}{1+\frac{\varepsilon-\con-\lin}{2\lin} + \sqrt{\left(1+\frac{\varepsilon-\con-\lin}{2\lin}\right)^2+\frac{4}{\pi \lin}}} \end{array} \right..$$
\end{lem}
\begin{proof}
Fix neighbouring $x,x' \in \mathcal{X}^n$ and let $E$ and $E'$ be the events promised by definition \ref{defn:ApproxRDP}. We can assume, without loss of generality that $\pr{}{E}=\pr{}{E'}=1-\delta$. 

Fix $S \subset \mathcal{Y}$. Then 
\begin{align*}
\pr{}{M(x) \in S} 
=& \pr{}{M(x) \in S \mid E} \cdot \pr{}{E} + \pr{}{M(x) \in S \mid \neg E} \cdot \pr{}{\neg E}\\
\leq& \pr{}{M(x) \in S \mid E} \cdot (1-\delta) + \delta,\\
\pr{}{M(x') \in S} =& \pr{}{M(x') \in S \mid E'} \cdot \pr{}{E'} + \pr{}{M(x') \in S \mid \neg E'} \cdot \pr{}{\neg E'}\\
\geq& \pr{}{M(x') \in S \mid E'} \cdot (1-\delta).
\end{align*}
Firstly, if $\lin=0$, then $$\pr{}{M(x) \in S \mid E} \leq e^\con \pr{}{M(x') \in S \mid E'}$$
and $$\pr{}{M(x) \in S} \leq \pr{}{M(x) \in S \mid E} \cdot (1-\delta) + \delta \leq e^\con \pr{}{M(x') \in S \mid E'} \cdot (1-\delta) + \delta \leq e^\con \pr{}{M(x') \in S} + \delta,$$ which proves the first half of the lemma.

Secondly, by Lemma \ref{lem:RenyiToED} (cf.~Lemma \ref{lem:CDPtoDP}), for all $\varepsilon \geq \con+\lin$, $$\pr{}{M(x) \in S \mid E} \leq e^\varepsilon \pr{}{M(x') \in S \mid E'} + e^{-(\varepsilon-\con-\lin)^2/4\lin} \cdot \min \left\{ \begin{array}{l} 1 \\ \sqrt{\pi \cdot \lin} \\ \frac{1}{1+(\varepsilon-\con-\lin)/2\lin} \\ \frac{2}{1+\frac{\varepsilon-\con-\lin}{2\lin} + \sqrt{\left(1+\frac{\varepsilon-\con-\lin}{2\lin}\right)^2+\frac{4}{\pi \lin}}} \end{array} \right..$$
Thus $$\pr{}{M(x) \in S} \leq e^\varepsilon \pr{}{M(x') \in S} + \delta + (1-\delta) \cdot e^{-(\varepsilon-\con-\lin)^2/4\lin} \cdot \min \left\{ \begin{array}{l} 1 \\ \sqrt{\pi \cdot \lin} \\ \frac{1}{1+(\varepsilon-\con-\lin)/2\lin} \\ \frac{2}{1+\frac{\varepsilon-\con-\lin}{2\lin} + \sqrt{\left(1+\frac{\varepsilon-\con-\lin}{2\lin}\right)^2+\frac{4}{\pi \lin}}} \end{array} \right..$$
\end{proof}

\subsection{Application of Approximate \IIP{}}

Approximate \IIP{} subsumes approximate DP. A result of this is that we can apply our tightened lemmas to give a tighter version of the so-called advanced composition theorem \cite{DworkRV10}.

Note that the following results are subsumed by the bounds of Kairouz, Oh, and Viswanath \cite{KairouzOV15} and Murtagh and Vadhan \cite{MurtaghV16}. However, these bounds may be extended to analyse the composition of mechanisms satisfying \cdp{} with mechanisms satisfying approximate DP. We believe that such a ``unified'' analysis of composition will be useful.

Applying Corollary \ref{cor:DPtoApproxGDP}, Lemma \ref{lem:CompPost-ApproxGDP}, and Lemma \ref{lem:ApproxGDPtoDP} yields the following result.

\begin{cor}
Let $M_1, \cdots, M_k : \mathcal{X}^n \to \mathcal{Y}$ and let $M : \mathcal{X}^n \to \mathcal{Y}^k$ be their composition. Suppose each $M_i$ satisfies $(\varepsilon_i,\delta_i)$-DP. Set $\lin = \frac12 \sum_i^k \varepsilon_i^2$. Then $M$ satisfies $$\left(\varepsilon, 1-(1-\delta')\prod_i^k (1-\delta_i)\right)\text{-DP}$$ for all $\varepsilon \geq \lin$ and $$\delta' = e^{-(\varepsilon-\lin)^2/4\lin} \cdot \min \left\{ \begin{array}{l} 1 \\ \sqrt{\pi \cdot \lin} \\ \frac{1}{1+(\varepsilon-\lin)/2\lin} \\ \frac{2}{1+\frac{\varepsilon-\lin}{2\lin} + \sqrt{\left(1+\frac{\varepsilon-\lin}{2\lin}\right)^2+\frac{4}{\pi \lin}}} \end{array} \right..$$
\end{cor}
A slight restatement is the following
\begin{cor}
Let $M_1, \cdots, M_k : \mathcal{X}^n \to \mathcal{Y}$ and let $M : \mathcal{X}^n \to \mathcal{Y}^k$ be their composition. Suppose each $M_i$ satisfies $(\varepsilon_i,\delta_i)$-DP. Set $\varepsilon^2 = \frac12 \sum_i^k \varepsilon_i^2$. Then $M$ satisfies $$\left(\varepsilon^2+2\lambda\varepsilon, 1-(1-\delta')\prod_i^k (1-\delta_i)\right)\text{-DP}$$ for all $\lambda \geq 0$ and $$\delta' = e^{-\lambda^2} \cdot \min \left\{ \begin{array}{l} 1 \\ \sqrt{\pi} \cdot \varepsilon \\ \frac{1}{1+\lambda/\varepsilon} \\ \frac{2}{1+\frac{\lambda}{\varepsilon} + \sqrt{\left(1+\frac{\lambda}{\varepsilon}\right)^2+\frac{4}{\pi \varepsilon^2}}} \end{array} \right..$$
\end{cor}
Finally, by picking the second term in the minimum and using $1-\prod_i(1-\delta_i) \leq \sum_i \delta_i$, we have the following simpler form of the lemma.
\begin{cor}
Let $M_1, \cdots, M_k : \mathcal{X}^n \to \mathcal{Y}$ and let $M : \mathcal{X}^n \to \mathcal{Y}^k$ be their composition. Suppose each $M_i$ satisfies $(\varepsilon_i,\delta_i)$-DP. Then $M$ satisfies $$\left(\frac12\|\varepsilon\|_2^2+\sqrt{2}\lambda\|\varepsilon\|_2, \sqrt{\frac{\pi}{2}} \cdot \|\varepsilon\|_2 \cdot e^{-\lambda^2} + \| \delta \|_1 \right)\text{-DP}$$ for all $\lambda \geq 0$. Alternatively $M$ satisfies $$\left(\frac12\|\varepsilon\|_2^2+\sqrt{2\log(\sqrt{\pi/2} \cdot \|\varepsilon\|_2/\delta')}\cdot\|\varepsilon\|_2,\delta' + \| \delta \|_1 \right)\text{-DP}$$ for all $\delta' \geq 0$.
\end{cor}

In comparison to the composition theorem of \cite{DworkRV10}, we save modestly by a constant factor in the first term and, in most cases $\sqrt{\pi/2} \|\varepsilon\|_2 < 1$, whence the logarithmic term is an improvement over the usual advanced composition theorem.

\addcontentsline{toc}{section}{Acknowledgements}
\paragraph{Acknowledgements}
We thank Cynthia Dwork and Guy Rothblum for sharing a preliminary draft of their work with us. We also thank Ilya Mironov, Kobbi Nissim, Adam Smith, Salil Vadhan, and the Harvard Differential Privacy Research Group for helpful discussions and suggestions.

\addcontentsline{toc}{section}{References}
\bibliographystyle{alpha}
\bibliography{refs.bib}

\appendix

\section{Postprocessing and \mcdp{}} \label{app:pp-mcdp}

In this appendix we give a family of counterexamples showing that \mcdp{} is \emph{not} closed under postprocessing (unlike \IIP{}).

Fix a parameter $\sigma > 0$, and consider the Gaussian mechanism for a single bit $M : \{-1, 1\} \to \mathbb{R}$, where $M(x)$ samples from $\mathcal{N}(x, \sigma^2)$. The mechanism $M$ satisfies $(2/\sigma^2, 2/\sigma)$-\mcdp{} (and also $(2/\sigma^2)$-\IIP{}).

Now consider the postprocessing function $T: \mathbb{R} \to \{-1, 0, 1\}$ defined as follows:
\[T(y) = \begin{cases}
1 & \text{ if } y > t \\
-1 & \text{ if } y < -t \\
0 & \text{ if } -t \le y \le t.
\end{cases}\]

We examine the \mcdp{} guarantees of the postprocessed mechanism $M' : \{-1, 1\} \to \{-1,0,1\}$ defined by $M'(x) = T(M(x))$:

\begin{prop}
Let $\sigma \ge 1$ and let $t \ge 6\sigma^3 + 1$. Then while the mechanism $M$ is $(2/\sigma^2, 2/\sigma)$-\mcdp{}, the postprocessed mechanism $M'$ is not $(2/\sigma^2, 2/\sigma)$-\mcdp{}.
\end{prop}

\begin{proof}
For each $x \in \{-1, 1\}$, let
\begin{align*}
p &= \pr{}{M'(x) = x} = \pr{}{\mathcal{N}(0, \sigma^2) > t - 1}, \\
q &= \pr{}{M'(x) = -x} = \pr{}{\mathcal{N}(0, \sigma^2) > t + 1}.
\end{align*}
Note that $p>q$. 
Hence, for each $x \in \{-1, 1\}$,
\[\pr{}{M'(x) = 0} = \pr{}{\mathcal{N}(0, \sigma^2) \in [-t-1, t-1]} = 1 - p - q.\]
Let $f(y) = \log(\pr{}{M(1) = y} / \pr{}{M(-1) = y})$, and observe that
\[f(1) = \log \frac{p}{q}, \qquad f(-1) = \log \frac{q}{p}, \qquad f(0) = 0.\]
Now consider the privacy loss random variable $Z = \privloss{M(1)}{M(-1)} = f(M(1))$. Then $Z$ is distributed according to
\[Z = \begin{cases}
\log \frac{p}{q} & \text{ w.p. } p \\
\log \frac{q}{p} & \text{ w.p. } q \\
0 & \text{ w.p. } 1-p-q
\end{cases}.\]
This gives $\ex{}{Z} = (p-q)\log(p/q) \geq 0$. For $\lambda \in \mathbb{R}$, we have
$$\ex{}{e^{\lambda(Z - \ex{}{Z})}} = \left( p \left(\frac{p}{q} \right)^\lambda + q \left(\frac{q}{p} \right)^\lambda + 1 - p - q \right) \cdot \left( \frac{p}{q} \right)^{-\lambda(p-q)}.$$

If $M'$ were to satisfy $(2/\sigma^2,2/\sigma)$-\mcdp{}, we would have $\ex{}{e^{\lambda(Z - \ex{}{Z})}} \leq e^{2\lambda^2/\sigma^2}$ for all $\lambda > 0$. We will show that this does not hold for any setting of parameters $\sigma \ge 1$ and $t \ge 6\sigma^3 + 1$, which shows that \mcdp{} is not closed under postprocessing.\footnote{For specific settings of parameters, this can be verified numerically. (For example, with the values $\sigma=1$, $t=3$, and $\lambda=2$.)}

\begin{lem} \label{lem:p-q-ineq}
The values $p, q$ satisfy the following inequalities:
\begin{enumerate}
\item $\sqrt{\frac{1}{2\pi}}\cdot \frac{\sigma}{t + \sigma -1} \cdot e^{-(t-1)^2/2\sigma^2} \le p \le \sqrt{\frac{1}{2\pi}} \cdot \frac{\sigma}{t-1} \cdot e^{-(t-1)^2/2\sigma^2}$
\item $\frac{p}{q} \ge e^{2t/\sigma^2}$
\end{enumerate}
\end{lem}

\begin{proof} 
We have \cite[Equation (5)]{Cook09}
$$\frac{\sqrt{\frac{2}{\pi}}e^{-x^2/2} \cdot x}{x^2+1} \leq \pr{}{\mathcal{N}(0,1)>x} \leq \frac{\sqrt{\frac{2}{\pi}}e^{-x^2/2}}{x} $$ for all $x \geq 0$. Thus $$\frac{\sqrt{\frac{2}{\pi}}e^{-(t-1)^2/2\sigma^2} \cdot (t-1)}{\sigma \cdot ((t-1)^2/\sigma^2 + 1)} \leq p=\pr{}{\mathcal{N}(0,1)>\frac{t-1}{\sigma}} \leq \frac{\sqrt{\frac{2}{\pi}}e^{-(t-1)^2/2\sigma^2} \cdot \sigma}{t-1} .$$
The first part now follows from the fact that $\sigma \le t-1$.

To establish the second inequality, we write 
\begin{align*}
p &= \frac{1}{\sqrt{2\pi\sigma^2}} \int_{t-1}^\infty e^{-x^2 / 2\sigma^2} \ dx \\
&= \frac{1}{\sqrt{2\pi\sigma^2}} \int_{t+1}^\infty e^{-(u-2)^2 / 2\sigma^2} \ du \\
&= \frac{1}{\sqrt{2\pi\sigma^2}} \int_{t+1}^\infty e^{(4u-4)/2\sigma^2} \cdot e^{-u^2 / 2\sigma^2} \ du \\
&\ge e^{2t/\sigma^2} \frac{1}{\sqrt{2\pi\sigma^2}} \int_{t+1}^\infty e^{-u^2 / 2\sigma^2} \ du \\
&= e^{2t/\sigma^2} \cdot q.
\end{align*}
\end{proof}

By Lemma \ref{lem:p-q-ineq},
\begin{align*}
\ex{}{e^{\lambda(Z - \ex{}{Z})}} &= \left( p \left(\frac{p}{q} \right)^\lambda + q \left(\frac{q}{p} \right)^\lambda + 1 - p - q \right) \cdot \left( \frac{p}{q} \right)^{-\lambda(p-q)}\\
&\ge p\left(\frac{p}{q} \right)^{\lambda(1 -(p - q))} \\
&\ge \sqrt{\frac{1}{2\pi}} \cdot \frac{\sigma}{t + \sigma - 1} \cdot e^{-(t-1)^2 / 2\sigma^2} \cdot (e^{2t/\sigma^2})^{\lambda(1 - p)} .
\end{align*}
Now set $\lambda = \frac{t}{2}$. Then this quantity becomes
\begin{equation} \label{eqn:fts}
\sqrt{\frac{1}{2\pi}} \cdot \frac{\sigma}{t + \sigma - 1} \cdot e^{-(t-1)^2 / 2\sigma^2} \cdot e^{t^2(1-p)/\sigma^2} =   e^{t^2/2\sigma^2} \cdot \sqrt{\frac{1}{2\pi}} \cdot \frac{\sigma}{t + \sigma - 1} \cdot \exp\left(\frac{t}{\sigma^2} - \frac{pt^2}{\sigma^2} - \frac{1}{2\sigma^2}\right).
\end{equation}
We now examine the term
\begin{align*}
pt^2 &\le \sqrt{\frac{1}{2\pi}}\frac{\sigma t^2e^{-(t-1)^2/2\sigma^2}}{t-1} & \text{by Lemma \ref{lem:p-q-ineq}} \\
&\le \sqrt{\frac{1}{4\pi}}\frac{t^{5/2} e^{-(t-1)}}{t-1} & \text{for } t \ge 2\sigma^2 + 1 \\
&\le \frac12 & \text{for } t \ge 3.
\end{align*}
We may thus bound (\ref{eqn:fts}) from below by
\begin{align*}
e^{t^2/2\sigma^2} \cdot \sqrt{\frac{1}{2\pi}} \cdot \frac{\sigma}{t + \sigma - 1} \cdot \exp \left(\frac{t - 1}{2\sigma^2} \right) &\ge e^{t^2/2\sigma^2} \cdot \sqrt{\frac{1}{2\pi}} \cdot \frac{1}{2t} \cdot \exp \left(\frac{t - 1}{2\sigma^2} \right) \\
\end{align*}
The expression $\exp((t-1)/2\sigma^2)/2t$ is monotone increasing for $t \ge 2\sigma^2$. Thus, it is strictly larger than $\sqrt{2\pi}$ as long as $t \ge 6\sigma^3 + 1$.
\end{proof}

\section{Miscellaneous Proofs and Lemmata}

\begin{lem} \label{lem:HyperTrigIneq}
$$0 \leq y < x \leq 2 \implies \frac{\sinh(x)-\sinh(y)}{\sinh(x-y)} \leq e^{\frac12 xy}.$$
\end{lem}

This technical lemma may be ``verified'' numerically by inspecting a plot of $z=e^{\frac12 xy} \cdot \sinh(x-y) - (\sinh(x)-\sinh(y))$ for $(x,y) \in [0,2]^2$.

For intuition, consider the third-order Taylor approximation to $\sinh(x)$ about $0$:
\[\sinh(x) = x + \frac{1}{6}x^3 \pm O(x^5).\]
Then we can approximate
\begin{align*}
\frac{\sinh(x) - \sinh(y)}{\sinh(x - y)} &\approx \dfrac{x + \frac{1}{6}x^3 - (y + \frac{1}{6}y^3)}{(x-y) + \frac{1}{6}(x-y)^3} \\
&= 1 + \frac{3xy(x-y)}{6(x-y) + (x-y)^3} \\
&\le 1 + \frac{1}{2}xy \\
&\le e^{\frac{1}{2}xy}.
\end{align*}
Unfortunately, turning this intuition into an actual proof is quite involved. We instead provide a proof from \cite{TrigMathSE}:

\begin{proof}
We need the following hyperbolic trigonometric identities.
\begin{align*}
\cosh\left(w+z\right) =& \cosh(w)\cosh(z) + \sinh(w)\sinh(z)\\
=& \cosh(w)\cosh(z) \left( 1 + \tanh(w)\tanh(z) \right),\\
\sinh(x-y) =& \frac12 \left( e^{x-y} + 1 - 1 - e^{-(x-y)} \right)\\
=& \frac12 \left( e^{(x-y)/2} - e^{-(x-y)/2} \right) \left( e^{(x-y)/2} + e^{-(x-y)/2} \right)\\
=& 2 \sinh\left(\frac{x-y}{2}\right) \cosh\left(\frac{x-y}{2}\right) \\
=& 2 \sinh\left(\frac{x-y}{2}\right) \cosh\left(\frac{x}{2}\right)\cosh\left(\frac{-y}{2}\right) \left( 1 + \tanh\left(\frac{x}{2}\right)\tanh\left(\frac{-y}{2}\right)  \right)\\
=& 2 \sinh\left(\frac{x-y}{2}\right) \cosh\left(\frac{x}{2}\right)\cosh\left(\frac{y}{2}\right) \left( 1 - \tanh\left(\frac{x}{2}\right)\tanh\left(\frac{y}{2}\right)  \right), \\
\sinh(x)-\sinh(y) =& \frac12 \left( e^x - e^{-x} - e^y + e^{-y} \right)\\
=& \frac12 \left( e^{(x-y)/2} - e^{-(x-y)/2} \right) \left( e^{(x+y)/2} + e^{-(x+y)/2} \right)\\
=& 2 \sinh\left(\frac{x-y}{2}\right) \cosh\left(\frac{x+y}{2}\right) \\
=& 2 \sinh\left(\frac{x-y}{2}\right) \cosh\left(\frac{x}{2}\right)\cosh\left(\frac{y}{2}\right) \left( 1 + \tanh\left(\frac{x}{2}\right)\tanh\left(\frac{y}{2}\right)  \right). \\
\end{align*}
We also use the fact that $0 \leq \tanh(z) < \min\{z,1\}$ for all $z > 0$.
Thus $$ \frac{\sinh(x)-\sinh(y)}{\sinh(x-y)} = \frac{1 + \tanh\left(\frac{x}{2}\right)\tanh\left(\frac{y}{2}\right) }{1 - \tanh\left(\frac{x}{2}\right)\tanh\left(\frac{y}{2}\right) } = \frac{1+t}{1-t},$$ where $t=\tanh\left(\frac{x}{2}\right)\tanh\left(\frac{y}{2}\right) < 1$. Now $$\frac{1+t}{1-t} \leq e^{xy/2} \iff 1+t \leq (1-t)e^{xy/2} 
\iff (e^{xy/2} + 1) t \leq e^{xy/2}-1$$ $$
\iff t \leq \frac{e^{xy/2}-1}{e^{xy/2}+1} = \frac{e^{xy/4}-e^{-xy/4}}{e^{xy/4}+e^{-xy/4}} = \tanh\left(\frac{xy}{4}\right).$$
So it only remains to show that $ \tanh\left(\frac{x}{2}\right)\tanh\left(\frac{y}{2}\right) \leq \tanh\left(\frac{xy}{4}\right)$. This is clearly true when $y=0$. Now $$\frac{\partial}{\partial y} \tanh\left(\frac{x}{2}\right) \tanh\left(\frac{y}{2}\right) = \tanh\left(\frac{x}{2}\right) \cosh^{-2}\left(\frac{y}{2}\right) \frac12 \leq  \cosh^{-2}\left(\frac{y}{2}\right) \frac{x}{4}$$
and
$$\frac{\partial}{\partial y} \tanh\left(\frac{xy}{4}\right) = \cosh^{-2}\left(\frac{xy}{4}\right) \frac{x}{4}.$$
Since $0 \leq y < x \leq 2$, we have $0 \leq xy/4 \leq y/2$ and, hence, $\cosh(y/2) \geq \cosh(xy/4)$. Thus $$\frac{\partial}{\partial y} \tanh\left(\frac{x}{2}\right) \tanh\left(\frac{y}{2}\right) \leq \frac{\partial}{\partial y} \tanh\left(\frac{xy}{4}\right) .$$
The inequality now follows by integration of the inequality on the derivatives.
\end{proof}

The following lemma immediately implies Lemma \ref{lem:CDPtoDP2} and is also used to prove Lemma \ref{lem:ApproxGDPtoDP}.

\begin{lem} \label{lem:RenyiToED}
Let $P$ and $Q$ be probability distributions on $\mathcal{Y}$ with $\dr{\alpha}{P}{Q} \leq \con + \lin \cdot \alpha$ for all $\alpha \in (1,\infty)$. Then, for any $\varepsilon\geq\con+\lin$ and $$\delta = e^{-(\varepsilon-\con-\lin)^2/4\lin} \cdot \min \left\{ \begin{array}{l} 1 \\ \sqrt{\pi \cdot \lin} \\ \frac{1}{1+(\varepsilon-\con-\lin)/2\lin} \\ \frac{2}{1+\frac{\varepsilon-\con-\lin}{2\lin} + \sqrt{\left(1+\frac{\varepsilon-\con-\lin}{2\lin}\right)^2+\frac{4}{\pi \lin}}} \end{array} \right.,$$
we have $$P(S) \leq e^\varepsilon Q(S) + \delta$$ for all (measurable) $S$.
\end{lem}
\begin{proof}
Define $f : \mathcal{Y} \to \mathbb{R}$ by $f(y) = \log(P(y)/Q(y))$. Let $Y\sim P$, $Y' \sim Q$ and let $Z=f(Y)$ be the privacy loss random variable. That is, $Z=\privloss{P}{Q}$.
For any measurable $S \subset \mathcal{Y}$, 
\begin{align*}
P(S) =& \pr{}{Y \in S}\\
=& \pr{}{Y \in S \wedge f(Y) \leq \varepsilon} +  \pr{}{Y \in S \wedge f(Y) > \varepsilon} \\
=& \int_{S} P(y) \mathbb{I}[f(y) \leq \varepsilon] \mathrm{d}y +  \pr{}{Y \in S \wedge f(Y) > \varepsilon} \\
=& \int_{S} P(y) \mathbb{I}[P(y) \leq e^\varepsilon Q(y)] \mathrm{d}y +  \pr{}{Y \in S \wedge f(Y) > \varepsilon} \\
\leq& \int_{S} e^\varepsilon Q(y) \mathbb{I}[f(y) \leq \varepsilon] \mathrm{d}y +  \pr{}{Y \in S \wedge f(Y) > \varepsilon} \\
=& e^\varepsilon \pr{}{Y' \in S \wedge f(Y') \leq \varepsilon} +  \pr{}{Y \in S \wedge f(Y) > \varepsilon} \\
=& e^\varepsilon \pr{}{Y' \in S} - e^\varepsilon \pr{}{Y' \in S \wedge f(Y') > \varepsilon} +  \pr{}{Y \in S \wedge f(Y) > \varepsilon} \\
\leq& e^\varepsilon \pr{}{Y' \in S} + \left( \pr{}{f(Y) > \varepsilon} - e^\varepsilon \pr{}{f(Y') > \varepsilon}  \right).
\end{align*}
Thus we want to bound 
\begin{align*}
\delta =&  \pr{}{f(Y) > \varepsilon} - e^\varepsilon \pr{}{f(Y') > \varepsilon}\\
=& \ex{Y}{\mathbb{I}[f(Y)>\varepsilon]} - \ex{Y'}{e^\varepsilon \mathbb{I}[f(Y')>\varepsilon]}\\
=& \ex{Y}{\mathbb{I}[f(Y)>\varepsilon]} - \int_{\mathcal{Y}} e^\varepsilon \mathbb{I}[f(y)>\varepsilon] Q(y) \mathrm{d}y\\
=& \ex{Y}{\mathbb{I}[f(Y)>\varepsilon]} - \int_{\mathcal{Y}} e^\varepsilon \mathbb{I}[f(y)>\varepsilon] \frac{Q(y)}{P(y)} P(y) \mathrm{d}y\\
=& \ex{Y}{\mathbb{I}[f(Y)>\varepsilon]} - \ex{Y}{e^\varepsilon \mathbb{I}[f(Y)>\varepsilon] e^{-f(Y)}}\\
=& \ex{Z}{\mathbb{I}[Z>\varepsilon] \left( 1 - e^{\varepsilon-Z} \right)}\\
=& \ex{Z}{\max \left\{ 0, 1 - e^{\varepsilon-Z} \right\}}\\
=& \int_\varepsilon^\infty \left( 1 - e^{\varepsilon-z} \right) \pr{}{Z=z} \mathrm{d}z.
\end{align*}
In particular, $$\delta \leq \ex{Z}{\mathbb{I}[Z>\varepsilon]} = \pr{}{Z>\varepsilon}.$$
Alternatively, by integration by parts, $$ \delta = \int_\varepsilon^\infty e^{\varepsilon-z} \pr{}{Z>z} \mathrm{d}z. $$
Now it remains to bound $\pr{}{Z>z}$.

As in Lemma \ref{lem:CDPtoDP}, by Markov's inequality, for all $\alpha>1$ and $\lambda>\con+\lin$, $$\pr{}{Z>\lambda} \leq \frac{\ex{}{e^{(\alpha-1) Z}}}{e^{(\alpha-1)\lambda}} \leq \frac{e^{(\alpha-1)\dr{\alpha}{M(x)}{M(x')}}}{e^{(\alpha-1)\lambda}} \leq e^{(\alpha-1)(\con+\lin\alpha-\lambda)}.$$ Choosing $\alpha = (\lambda - \con + \lin)/2\lin > 1$ gives $$\pr{}{Z>\lambda} \leq e^{-(\lambda - \con - \lin)^2/4\lin}.$$
Thus $\delta \leq \pr{}{Z>\varepsilon} \leq e^{-(\varepsilon-\con-\lin)^2/4\lin}$. Furthermore,
\begin{align*}
\delta =& \int_\varepsilon^\infty e^{\varepsilon-z} \pr{}{Z>z} \mathrm{d}z\\
\leq& \int_\varepsilon^\infty e^{\varepsilon-z} e^{-(z - \con - \lin)^2/4\lin} \mathrm{d}z\\
=& \int_\varepsilon^\infty e^{ -(z - \con + \lin)^2/4\lin - \con + \varepsilon} \mathrm{d}z\\
=& \frac{e^{\varepsilon-\con}\cdot\sqrt{2\pi \cdot 2\lin}}{\sqrt{2\pi \cdot 2\lin}}\int_\varepsilon^\infty e^{ -(z - \con + \lin)^2/4\lin} \mathrm{d}z\\
=& {e^{\varepsilon-\con}\cdot\sqrt{2\pi \cdot 2\lin}} \cdot \pr{}{\mathcal{N}(\con-\lin,2\lin) > \varepsilon}\\
=& {e^{\varepsilon-\con}\cdot2\sqrt{\pi \cdot \lin}} \cdot \pr{}{\mathcal{N}(0,1) > \frac{\varepsilon+\lin-\con}{\sqrt{2\lin}}}.
\end{align*}
Define $h(x) = \pr{}{\mathcal{N}(0,1) > x} \cdot \sqrt{2\pi} \cdot e^{x^2/2}$. Then $$\delta \leq {e^{\varepsilon-\con}\cdot2\sqrt{\pi \cdot \lin}} \cdot \frac{h\left(\frac{\varepsilon+\lin-\con}{\sqrt{2\lin}}\right)}{\sqrt{2\pi} \cdot e^{\left(\frac{\varepsilon+\lin-\con}{\sqrt{2\lin}}\right)^2/2}} = \sqrt{2\lin} \cdot e^{-(\varepsilon-\con-\lin)^2/4\lin} \cdot h\left(\sqrt{2\lin} + \frac{\varepsilon-\con-\lin}{\sqrt{2\lin}}\right).$$
Now we can substitute upper bounds on $h$ to obtain the desired results.

First we use $\pr{}{\mathcal{N}(0,1)>x} \leq \frac12 e^{-x^2/2}$, which holds for all $x \geq 0$. That is, $h(x)\leq\sqrt{\pi/2}$, whence
$$\delta \leq \sqrt{2\lin} \cdot e^{-(\varepsilon-\con-\lin)^2/4\lin} \cdot \sqrt{\frac{\pi}{2}}.$$
This rearranges to $$\varepsilon \leq \con+\lin+\sqrt{4\lin\cdot\log(\sqrt{\pi \cdot \lin} / \delta)}. $$

Alternatively, we can use $\pr{}{\mathcal{N}(0,1)>x} \leq e^{-x^2/2} / \sqrt{2\pi} x$, which holds for all $x > 0$. That is, $h(x)\leq 1/x$ and
$$\delta \leq \sqrt{2\lin} \cdot e^{-(\varepsilon-\con-\lin)^2/4\lin} \cdot h\left( \frac{\varepsilon-\con+\lin}{\sqrt{2\lin}}\right) = \frac{2\lin}{\varepsilon+\lin-\con} \cdot e^{-(\varepsilon-\con-\lin)^2/4\lin} =  \frac{e^{-(\varepsilon-\con-\lin)^2/4\lin}}{1+(\varepsilon-\con-\lin)/2\lin} \cdot .$$

Finally, we can use $\pr{}{\mathcal{N}(0,1)>x} \leq e^{-x^2/2} \cdot 2 /(\sqrt{x^2+8/\pi}+x)\sqrt{2\pi}$, which holds for all $x\geq0$ \cite[Equation (4)]{Duembgen10}\cite[Equation (7)]{Cook09}\cite[Equation 7.1.13]{AbramowitzS64}. 
That is, $h(x) \leq 2/(\sqrt{x^2+8/\pi}+x)$ and $$\delta \leq  \sqrt{2\lin} \cdot e^{-(\varepsilon-\con-\lin)^2/4\lin} \cdot h\left(\frac{\varepsilon-\con+\lin}{\sqrt{2\lin}}\right) = \frac{2\cdot e^{-(\varepsilon-\con-\lin)^2/4\lin}}{1+\frac{\varepsilon-\con-\lin}{2\lin} + \sqrt{\left(1+\frac{\varepsilon-\con-\lin}{2\lin}\right)^2+\frac{4}{\pi \lin}}}.$$
\end{proof}

\begin{lem}[Restating Lemma \ref{lem:DPtoCDP}] \label{lem:DPtoCDP-app}
Let $M : \mathcal{X}^n \to \mathcal{Y}$ satisfy $(\varepsilon,\delta)$-DP for all $\delta>0$ and \begin{equation} \varepsilon = \hat \con  +\sqrt{\hat \lin\log(1/\delta)}\label{eqn:DP}\end{equation} for some constants $\hat \con, \hat \lin \in [0,1]$. Then $M$ is $\left(\hat \con-\frac{1}{4} \hat \lin+5\sqrt[4]{\hat \lin}, \frac{1}{4} \hat \lin\right)$-\IIP{}.
\end{lem}
\begin{proof}
Let $x,x' \in \mathcal{X}^n$ be neighbouring. Define $f(y) = \log(\pr{}{M(x)=y}/\pr{}{M(x')=y})$. Let $Y\sim M(x)$ and $Z=f(Y)$. That is, $Z=\privloss{M(x)}{M(x')}$ is the privacy loss random variable. Let $Y' \sim M(x')$ and $Z'=f(Y')$. That is, $-Z'$ is the privacy loss random variable if we swap $x$ and $x'$.

Let $\varepsilon,\delta>0$ satisfy \eqref{eqn:DP}.
By postprocessing, for all $t \in \mathbb{R}$,
\begin{align*}
\pr{}{Z>t} =& \pr{}{f(M(x)) > t}\\
\leq& e^\varepsilon \pr{}{f(M(x')) > t} + \delta\\
=& e^\varepsilon \int_\mathcal{Y} \pr{}{M(x')=y} \cdot \mathbb{I}(\pr{}{M(x)=y} > e^t \pr{}{M(x')=y} )\mathrm{d}y + \delta \\
<& e^\varepsilon \int_\mathcal{Y} \pr{}{M(x)=y} \cdot e^{-t} \cdot \mathbb{I}(\pr{}{M(x)=y} > e^t \pr{}{M(x')=y} )\mathrm{d}y + \delta\\
=& e^{\varepsilon - t} \pr{}{Z>t} +\delta,
\end{align*}
whence $\pr{}{Z>t} \leq \frac{\delta}{1-e^{\varepsilon-t}}$. Then we can set $\varepsilon =\hat \con+\lambda$ and $\delta = e^{-\lambda^2/\hat \lin}$ to obtain $$\forall t>0 ~~~ \pr{}{Z>\hat \con+t} \leq \inf_{0<\lambda<t} \frac{e^{-\lambda^2/\hat \lin}}{1-e^{\lambda-t}}.$$

In particular, for all $t\geq 0$, $$\pr{}{Z>\hat \con + \sqrt[4]{\hat \lin} + t} \leq \frac{e^{-t^2/\hat \lin}}{1-e^{-\sqrt[4]{\hat \lin}}} \leq \frac{2}{\sqrt[4]{\hat \lin}} e^{-t^2/\hat \lin}.$$

We use the inequality $1+xe^\con \leq e^{x+\con}$ for all $x,\con \geq 0$.
We have
\begin{align*}
\ex{}{e^{(\alpha-1)Z}}  =& \int_0^\infty \pr{}{e^{(\alpha-1)Z}>t} \mathrm{d}t\\
 =& \int_0^\infty \pr{}{Z>\frac{\log t}{\alpha-1}} \mathrm{d}t\\
 =& \int_{-\infty}^\infty \pr{}{Z>z} \frac{\mathrm{d}t}{\mathrm{d}z}\mathrm{d}z\\
=& \int_{-\infty}^\infty (\alpha-1)e^{(\alpha-1)z}\cdot \pr{}{Z>z}\mathrm{d}z\\
\leq& \int_{-\infty}^{\hat \con + \sqrt[4]{\hat \lin}} (\alpha-1)e^{(\alpha-1)z}\cdot 1\mathrm{d}z + \int_{0}^{\infty} (\alpha-1)e^{(\alpha-1)(t+{\hat \con + \sqrt[4]{\hat \lin}})}\cdot  \frac{2}{\sqrt[4]{\hat \lin}} e^{-t^2/\hat \lin}\mathrm{d}t\\
=& e^{(\alpha-1)({\hat \con + \sqrt[4]{\hat \lin}})} + (\alpha-1)e^{(\alpha-1)({\hat \con + \sqrt[4]{\hat \lin}})}\frac{2}{\sqrt[4]{\hat \lin}}\int_{0}^{\infty} e^{(\alpha-1)t}\cdot   e^{-t^2/\hat \lin}\mathrm{d}t\\
\leq& e^{(\alpha-1)({\hat \con + \sqrt[4]{\hat \lin}})} + (\alpha-1)e^{(\alpha-1)({\hat \con + \sqrt[4]{\hat \lin}})}\frac{2}{\sqrt[4]{\hat \lin}}\int_{-\infty}^{\infty} e^{-(t-\frac12(\alpha-1){\hat \lin})^2/\hat \lin + \frac{1}{4}(\alpha-1)^2 \hat \lin}\mathrm{d}t\\
=& e^{(\alpha-1)({\hat \con + \sqrt[4]{\hat \lin}})} + (\alpha-1)e^{(\alpha-1)({\hat \con + \sqrt[4]{\hat \lin}})}\frac{2}{\sqrt[4]{\hat \lin}}  e^{\frac{1}{4}(\alpha-1)^2 \hat \lin}\sqrt{\pi \hat \lin}\\
=& e^{(\alpha-1)({\hat \con + \sqrt[4]{\hat \lin}})} \left( 1+ (\alpha-1){2}  \sqrt{\pi} \sqrt[4]{\hat \lin} e^{\frac{1}{4}(\alpha-1)^2 \hat \lin}\right)\\
\leq& e^{(\alpha-1)({\hat \con + \sqrt[4]{\hat \lin}})} e^{ (\alpha-1){2}  \sqrt{\pi} \sqrt[4]{\hat \lin}+ {\frac{1}{4}(\alpha-1)^2 \hat \lin}}\\
=& e^{(\alpha-1)\left({\hat \con + \sqrt[4]{\hat \lin}}+{2}  \sqrt{\pi} \sqrt[4]{\hat \lin} + {\frac{1}{4}(\alpha-1) \hat \lin}\right)}\\
=& e^{(\alpha-1)\left(\hat \con + (1+2\sqrt{\pi})\sqrt[4]{\hat \lin} - \frac{\hat \lin}{4} + \frac{\hat \lin}{4}\alpha\right)}.
\end{align*}
Since $\ex{}{e^{(\alpha-1)Z}} = e^{(\alpha-1)\dr{\alpha}{M(x)}{M(x')}}$, this completes the proof.
\end{proof}

We make use of the following technical lemma taken from \cite{DworkSSUV15}.
\begin{lem} \label{lem:Symm}
Let $X$ be a random variable. Then $$\ex{}{e^{X-\ex{}{X}}} \leq \frac12 \ex{}{e^{2X}} + \frac12 \ex{}{e^{-2X}}.$$
\end{lem}
\begin{proof}
Let $X'$ be an independent copy of $X$ and let $Y \in \{0,1\}$ be uniformly random and independent from $X$ and $X'$. By Jensen's inequality,
\begin{align*}\ex{}{e^{X-\ex{}{X}}} =& \ex{X}{e^{\ex{X'}{X-X'}}} \leq  \ex{X, X'}{e^{{X-X'}}} = \ex{X, X'}{e^{2\ex{Y}{YX-(1-Y)X'}}} \\\leq&  \ex{X, X', Y}{e^{{2YX-2(1-Y)X'}}}
= \frac12 \ex{X,X'}{e^{2X-0}} + \frac12 \ex{X,X'}{e^{0-2X'}} = \frac12 \ex{}{e^{2X}} + \frac12 \ex{}{e^{-2X}}.
\end{align*}
\end{proof}
\begin{lem}[Restating Lemma \ref{lem:IIPtoCDP}]
If $M : \mathcal{X}^n \to \mathcal{Y}$ satisfies $(\con,\lin)$-\IIP{}, then $M$ satisfies $(\con+\lin,O(\sqrt{\con+2\lin}))$-CDP.
\end{lem}
\begin{proof}
Let $x$ and $x'$ be neighbouring databases and $Z\sim\privloss{M(x)}{M(x')}$ the privacy loss random variable
We have $\ex{}{Z}=\dr{1}{M(x)}{M(x')} \in [0, \con+\lin]$ by non-negativity and \IIP{}.
By \IIP{}, for all $\alpha \in (1,\infty)$, we have $$\ex{}{e^{(\alpha-1)Z}} = e^{(\alpha-1)\dr{\alpha}{M(x)}{M(x')}} \leq e^{(\alpha-1)(\con+\lin\alpha)}$$ and
\begin{align*}
\ex{}{e^{-\alpha Z}} =& \ex{Y \sim M(x)}{\left(\frac{\pr{}{M(x)=Y}}{\pr{}{M(x')=Y}}\right)^{-\alpha}}\\
=& \ex{Y \sim M(x)}{\left(\frac{\pr{}{M(x')=Y}}{\pr{}{M(x)=Y}}\right)^{\alpha}}\\
=& e^{(\alpha-1)\dr{\alpha}{M(x')}{M(x)}}\\
\leq& e^{(\alpha-1)(\con+\lin\alpha)}.
\end{align*}
By Lemma \ref{lem:Symm}, for $\lambda \geq 1/2$, $$\ex{}{e^{\lambda (Z-\ex{}{Z})}} \leq \frac12 \ex{}{e^{2\lambda Z}} + \frac12 \ex{}{e^{-2\lambda Z}} \leq \frac12 e^{2\lambda(\con+\lin(2\lambda+1))} + \frac12 e^{(2\lambda-1)(\con+2\lin\lambda)} \leq e^{4(\con+2\lin)\lambda^2}$$ and, for $\lambda \leq -1/2$, $$\ex{}{e^{\lambda (Z-\ex{}{Z})}} \leq \frac12 \ex{}{e^{2\lambda Z}} + \frac12 \ex{}{e^{-2\lambda Z}} \leq \frac12 e^{(-2\lambda-1)(\con-2\lin\lambda)} + \frac12 e^{-2\lambda(\con+\lin(1-2\lambda))}  \leq e^{4(\con+2\lin)\lambda^2}.$$
Now suppose $|\lambda|< 1/2$. Then
\begin{align*}
\ex{}{e^{\lambda (Z-\ex{}{Z})}} \leq&  \frac12 \ex{}{e^{2\lambda Z}} + \frac12 \ex{}{e^{-2\lambda Z}}\\
=& 1 + \sum_{k=1}^\infty \frac{(2\lambda)^{2k}}{(2k)!} \ex{}{Z^{2k}}\\
\leq& 1 + (2\lambda)^2\sum_{k=1}^\infty \frac{1}{(2k)!} \ex{}{Z^{2k}}\\
=& 1 + 4\lambda^2 \left(\frac12 \ex{}{e^Z} + \frac12 \ex{}{e^{-Z}} - 1  \right)\\
\leq& 1 + 4\lambda^2 \left(e^{\con + 2\lin} - 1  \right)\\
\leq& e^{\lambda^2O(\con+2\lin)}.
\end{align*}
\end{proof}

\subsection{Proof of Lemma \ref{lem:Renyi}} \label{app:Renyi}

\begin{proof}[Proof of Non-Negativity] Let $h(t)=t^\alpha$. Then $h''(t) = \alpha(\alpha-1)t^{\alpha-2} > 0$ for all $t>0$ and $\alpha > 1$. Thus $h$ is strictly convex. Hence $e^{(\alpha-1)\dr{\alpha}{P}{Q})} = \ex{x \sim Q}{h(P(x)/Q(x))} \geq h(\ex{x \sim Q}{P(x)/Q(x)})=h(1)=1$, as required. \end{proof}
\begin{proof}[Proof of Composition] \begin{align*}e^{(\alpha-1) \dr{\alpha}{P}{Q}} =& \int_{\Omega \times \Theta} P(x,y)^\alpha Q(x,y)^{1-\alpha} \mathrm{d}(x,y)\\ =& \int_\Omega P'(x)^\alpha Q'(x)^{1-\alpha} \int_\Theta P'_x(y)^\alpha Q'_x(y)^{1-\alpha} \mathrm{d}y \mathrm{d}x\\=& \int_\Omega P'(x)^\alpha Q'(x)^{1-\alpha} e^{(\alpha-1)\dr{\alpha}{P'_x}{Q'_x}} \mathrm{d}x\\\leq& \int_\Omega P'(x)^\alpha Q'(x)^{1-\alpha} \mathrm{d}x \cdot \max_x e^{(\alpha-1)\dr{\alpha}{P'_x}{Q'_x}} \\=& e^{(\alpha-1)\dr{\alpha}{P'}{Q'}} \cdot  e^{(\alpha-1) \max_x \dr{\alpha}{P'_x}{Q'_x}}. \end{align*}
The other side of the inequality is symmetric. \end{proof}

\begin{proof}[Proof of Quasi-Convexity]

Unfortunately R\'enyi divergence is not convex for $\alpha>1$. (Although KL-divergence is.) However, the following property implies that R\'enyi divergence is quasi-convex.

\begin{lem} \label{lem:ExpConvex}
Let $P_0,P_1,Q_0,Q_1$ be distributions on $\Omega$. For $t \in [0,1]$, define $P_t=t P_1 + (1-t)P_0$ and $Q_t=tQ_1+(1-t)Q_0$ to be the convex combinations specified by $t$. Then $$e^{(\alpha-1)\dr{\alpha}{P_t}{Q_t}} \leq t e^{(\alpha-1)\dr{\alpha}{P_1}{Q_1}} + (1-t)e^{(\alpha-1)\dr{\alpha}{P_0}{Q_0}}.$$ 
\end{lem}
Moreover,  the limit as $\alpha \to 1+$ gives $$\dr{1}{P_t}{Q_t} \leq t \dr{1}{P_1}{Q_1} + (1-t) \dr{1}{P_0}{Q_0}.$$

\begin{proof}[Proof of Lemma \ref{lem:ExpConvex}]
Let $f(t) = e^{(\alpha-1)\dr{\alpha}{P_t}{Q_t}}$. Since the equality is clearly true for $t=0$ and $t=1$, it suffices to show that $f''(t)\geq 0$ for all $t \in [0,1]$. We have
\begin{align*}
f(t) =& \int_\Omega P_t(x)^\alpha Q_t(x)^{1-\alpha} \mathrm{d}x,\\
f'(t) =& \int_\Omega \frac{\mathrm{d}}{\mathrm{d}t} P_t(x)^\alpha Q_t(x)^{1-\alpha} \mathrm{d}x\\
=& \int_\Omega \alpha P_t(x)^{\alpha-1}\left(\frac{\mathrm{d}}{\mathrm{d}t} P_t(x)\right) Q_t(x)^{1-\alpha} + (1-\alpha) P_t(x)^\alpha Q_t(x)^{-\alpha} \left(\frac{\mathrm{d}}{\mathrm{d}t} Q_t(x) \right)  \mathrm{d}x\\
=& \int_\Omega \alpha P_t(x)^{\alpha-1}\left(P_1(x)-P_0(x)\right) Q_t(x)^{1-\alpha} + (1-\alpha) P_t(x)^\alpha Q_t(x)^{-\alpha} \left(Q_1(x) - Q_0(x) \right)  \mathrm{d}x,\\
f''(t) =& \int_\Omega \alpha(\alpha-1) P_t(x)^{\alpha-2}\left(P_1(x)-P_0(x)\right)^2 Q_t(x)^{1-\alpha} \\& + \alpha(1-\alpha) P_t(x)^{\alpha-1}\left(P_1(x)-P_0(x)\right) Q_t(x)^{-\alpha} \left(Q_1(x) - Q_0(x) \right) \\& + (1-\alpha) \alpha P_t(x)^{\alpha-1} \left(P_1(x)-P_0(x)\right) Q_t(x)^{-\alpha} \left(Q_1(x) - Q_0(x) \right) \\& + (1-\alpha)(-\alpha) P_t(x)^\alpha Q_t(x)^{-\alpha-1} \left(Q_1(x) - Q_0(x) \right)^2  \mathrm{d}x\\
 =& \alpha(\alpha-1) \int_\Omega \left( \sqrt{P_t(x)^{\alpha-2} Q_t(x)^{1-\alpha}} \left(P_1(x)-P_0(x)\right) - \sqrt{P_t(x)^\alpha Q_t(x)^{-\alpha-1}} \left(Q_1(x) - Q_0(x) \right) \right)^2 \mathrm{d}x\\
\geq& 0.
\end{align*}
\end{proof}
\end{proof}

\begin{proof}[Proof of Postprocessing] Let $h(x)=x^\alpha$. Note that $h$ is convex. Let $f^{-1}(y) = \{x \in \Omega : f(x)=y\}$. Let $Q_y$ be the conditional distribution on $x \sim Q$ conditioned on $f(x)=y$. By Jensen's inequality,
\begin{align*}
e^{(\alpha-1)\dr{\alpha}{P}{Q}} =& \ex{x \sim Q}{\left(\frac{P(x)}{Q(x)}\right)^\alpha}\\
 =& \ex{y \sim f(Q)}{\ex{x \sim Q_y}{h\left(\frac{P(x)}{Q(x)}\right)}}\\
 \geq& \ex{y \sim f(Q)}{h\left(\ex{x \sim Q_y}{\frac{P(x)}{Q(x)}}\right)}\\
=& \ex{y \sim f(Q)}{h\left(\int_{f^{-1}(y)} \frac{Q(x)}{Q(f^{-1}(y))} \frac{P(x)}{Q(x)} \mathrm{d}x\right)}\\
=& \ex{y \sim f(Q)}{h\left(\frac{P(f^{-1}(y))}{Q(f^{-1}(y))}\right)}\\
=& e^{(\alpha-1)\dr{\alpha}{f(P)}{f(Q)}}.
\end{align*}
\end{proof}
\begin{proof}[Proof of Monotonicity]
Let $1 <\alpha \leq \alpha' < \infty$. Let $h(x)=x^{\frac{\alpha'-1}{\alpha-1}}$. Then $$h''(x) = \frac{\alpha'-1}{\alpha-1}\left(\frac{\alpha'-1}{\alpha-1}-1\right) x^{\frac{\alpha'-1}{\alpha-1}-2} \geq 0,$$ so $h$ is convex on $(0,\infty)$. Thus $$e^{(\alpha'-1)\dr{\alpha}{P}{Q}} = h\left(e^{(\alpha-1)\dr{\alpha}{P}{Q}}\right) = h\left(\ex{x \sim P}{\left(\frac{P(x)}{Q(x)}\right)^{\alpha-1}}\right) \leq \ex{x \sim P}{h\left(\left(\frac{P(x)}{Q(x)}\right)^{\alpha-1}\right)} = e^{(\alpha'-1)\dr{\alpha'}{P}{Q}},$$
which gives the result.
\end{proof}

\section{Privacy versus Sampling}

In this appendix we prove the technical Lemma \ref{lem:SubSampAcc}. Essentially we show that if a private mechanism can accurately answer a set of queries with a given sample complexity, then those queries can be approximated on an unknown distribution with the same sample complexity. This is related to lower bounds on private sample complexity using Vapnik-Chervonenkis dimension e.g. \cite[Theorem 4.8]{DworkR14} and \cite[Theorem 5.5]{BunNSV15}.

First we state Pinsker's inequality \cite[Theorem 31]{vanErvenH14}.

\begin{lem}[Pinsker's Inequality]
Let $X$ and $Y$ be random variables on $[-1,1]$. Then $$\left|\ex{}{X}-\ex{}{Y}\right| \leq \sqrt{2\dr{1}{X}{Y}} $$ 
\end{lem}

We can also generalise Pinsker's inequality using R\'{e}nyi divergence:

\begin{lem}
Let $P$ and $Q$ be distributions on $\Omega$ and $f : \Omega \to \mathbb{R}$. Then $$\left| \ex{x \sim P}{f(x)} - \ex{x \sim Q}{f(x)} \right| \leq \sqrt{\ex{x \sim Q}{f(x)^2}} \cdot \sqrt{e^{\dr{2}{P}{Q}}-1}.$$
\end{lem}
In particular, if $M : \mathcal{X}^n \to \mathcal{Y}$ satisfies $(\con,\lin)$-\IIP{}. Then, for any $f : \mathcal{Y} \to \mathbb{R}$ and all neighbouring $x,x' \in \mathcal{X}^n$, $$\left|\ex{}{f(M(x'))} - \ex{}{f(M(x))}\right| \leq \sqrt{\ex{}{f(M(x))^2}} \cdot \sqrt{e^{\con+2\lin}-1}. $$
\begin{proof}
By Cauchy-Schwartz,
\begin{align*}
\ex{x \sim P}{f(x)} - \ex{x \sim Q}{f(x)} =& \ex{x \sim Q}{f(x) \left(\frac{P(x)}{Q(x)} -1\right)}\\
\leq& \sqrt{\ex{x \sim Q}{f(x)^2}} \cdot \sqrt{\ex{x \sim Q}{\left(\frac{P(x)}{Q(x)} -1\right)^2}}.
\end{align*}
Now
\begin{align*}
\ex{x \sim Q}{\left(\frac{P(x)}{Q(x)} -1\right)^2}
=& \ex{x \sim Q}{\left(\frac{P(x)}{Q(x)}\right)^2 -2 \frac{P(x)}{Q(x)} + 1}\\
=& \ex{x \sim Q}{\left(\frac{P(x)}{Q(x)}\right)^2}-2+1\\
=& e^{\dr{2}{P}{Q}}-1. 
\end{align*}
\end{proof}

\begin{prop} \label{prop:Generalize}
Let $q : \mathcal{X} \to \mathbb{R}^k$ satisfy $\max_{x \in \mathcal{X}} \|q(x)\| \leq 1$, where $\|\cdot\|$ is some norm. Let $M : \mathcal{X}^n \to \mathbb{R}^k$ satisfy $(\con,\lin)$-\IIP{} and $$\ex{M}{\|M(x)-q(x)\|} \leq \alpha $$ for all $x \in \mathcal{X}^n$. Then, for any distribution $\mathcal{D}$ on $\mathcal{X}$, $$\ex{x \sim \mathcal{D}^n,M}{\|M(x) - q(\mathcal{D}) \|} \leq \alpha + \sqrt{2(\con+\lin)}$$ and $$\ex{x \sim \mathcal{D}^n}{\|q(x)-q(\mathcal{D})\|} \leq 2\alpha + \sqrt{2(\con+\lin)},$$ where $q(\mathcal{D}) = \ex{z \sim \mathcal{D}}{q(z)}$.
\end{prop}
\begin{proof}
First define the dual norm: For $x \in \mathbb{R}^k$, $$\|x\|_* := \max_{y \in \mathbb{R}^k : \|y\|=1} {\langle x, y \rangle}.$$ By definition, $\langle x, y \rangle = \langle y, x \rangle \leq \|x\|_* \cdot \|y\|$ for all $x,y \in \mathbb{R}^k$. Moreover, $\|z\| = \max_{y \in \mathbb{R}^k : \|y\|_*=1} \langle z, y \rangle$ for all $z \in \mathbb{R}^k$.

Fix a distribution $\mathcal{D}$. Define $W : \mathcal{X}^n \to \mathbb{R}^k \times \mathbb{R}^k$ as follows.
On input $x \in \mathcal{X}^n$, compute $a=M(x)$ and $$s = \underset{v \in \mathbb{R}^k : \|v\|_* = 1}{\mathrm{argmax}} \langle a-q(\mathcal{D}), v \rangle,$$ and output $(a,s)$.

By postprocessing, $W$ satisfies $(\con,\lin)$-\IIP{} and \begin{equation}\ex{W}{\langle a - q(x), s \rangle \mid (a,s) = W(x)} \leq \ex{W}{\| a - q(x)\| \cdot \|s \|_* \mid (a,s) = W(x)} = \ex{M}{\|M(x)-q(x)\|} \leq \alpha\label{eqn:EmpErr}\end{equation} for all $x \in \mathcal{X}^n$.

The following is similar to \cite[Lemma 3.1]{BassilyNSSSU16}. Let $x \sim \mathcal{D}^n$ and $y \sim \mathcal{D}$. Now
\begin{align}\label{eqn:GenErr}
\ex{x,W}{\langle q(x), s \rangle\mid (a,s) = W(x)} 
=& \frac{1}{n} \sum_{i \in [n]} \ex{x,W}{\langle q(x_i), s \rangle \mid (a,s) = W(x)}\\\nonumber
=& \frac{1}{n} \sum_{i \in [n]} \ex{x,W}{f(x_i, W(x))}\\\nonumber
\intertext{(letting $f : \mathcal{X} \times \mathbb{R}^k \times \mathbb{R}^k \to [-1,1]$ be $f(z,a,s) = \langle q(z), s\rangle/\|s\|_*$)}\nonumber
\leq& \frac{1}{n} \sum_{i \in [n]} \ex{x,y,W}{f(x_i, W(x_1, \cdots, x_{i-1}, y, x_{i+1}, \cdots, x_n))} \\\nonumber&+ \sqrt{2\dr{1}{f(x_i, W(x))}{f(x_i, W(x_1, \cdots, x_{i-1}, y, x_{i+1}, \cdots, x_n))}}\\\nonumber
\intertext{(by Pinsker's inequality)}\nonumber
\leq& \frac{1}{n} \sum_{i \in [n]} \ex{x,y,W}{f(y, W(x_1, \cdots, x_{i-1}, x_i, x_{i+1}, \cdots, x_n))} \\\nonumber&+ \sqrt{2\dr{1}{W(x)}{W(x_1, \cdots, x_{i-1}, y, x_{i+1}, \cdots, x_n)}}\\\nonumber
\intertext{(by postprocessing and convexity and the fact that $x_i$ and $y$ are interchangable)}\nonumber
\leq& \frac{1}{n} \sum_{i \in [n]} \ex{x,y,W}{\langle q(y), s\rangle \mid (j,s,a) = W(x)} \\\nonumber&+ \sqrt{2(\con+\lin)}\\\nonumber
\intertext{(by \IIP{})}\nonumber
=& \ex{x,W}{\langle q(\mathcal{D}), s \rangle \mid (a,s)=W(x)}+\sqrt{2(\con+\lin)}.
\end{align}
Combining \eqref{eqn:EmpErr} and \eqref{eqn:GenErr} gives
\begin{equation}\ex{x,M}{\|M(x) - q(\mathcal{D}) \|} =\ex{x,W}{\langle a-q(\mathcal{D}), s \rangle \mid (a,s) = W(x)} \leq \alpha + \sqrt{2(\con+\lin)}. \label{eqn:TotErr}\end{equation}
Finally, combining \eqref{eqn:TotErr} and \eqref{eqn:EmpErr} gives $$\ex{x}{\|q(x) - q(\mathcal{D}) \|} \leq \ex{x,M}{\|M(x) - q(x) \|} + \ex{x,M}{\|M(x) - q(\mathcal{D}) \|} \leq 2\alpha + \sqrt{2(\con+\lin)}.$$
\end{proof}

\begin{proof}[Proof of Lemma \ref{lem:SubSampAcc}] 
Fix $x \in \mathcal{X}^{n'}$. Let $\mathcal{D}$ be the uniform distribution on elements of $x$ so that $q(\mathcal{D})=q(x)$. By Proposition \ref{prop:Generalize}, $$\ex{\hat x \sim \mathcal{D}^n}{\|q(\hat x)-q(\mathcal{D})\|} \leq 2\alpha + \sqrt{2(\con+\lin)}.$$
In particular, there must exist $\hat x \sim \mathcal{D}^n$ such that $\|q(\hat x)-q(\mathcal{D})\| \leq 2\alpha + \sqrt{2(\con+\lin)}$, as required.
\end{proof}

\end{document}